\documentclass[reqno,12pt,letterpaper]{amsart}
\usepackage{amsmath,amssymb,amsthm,graphicx,mathrsfs,url}
\usepackage[usenames,dvipsnames]{color}
\usepackage[colorlinks=true,linkcolor=Red,citecolor=Green]{hyperref}
\usepackage{amsxtra}
\usepackage{placeins}
\usepackage{wasysym} 
\usepackage{graphicx}
\usepackage{subcaption}
\def\arXiv#1{\href{http://arxiv.org/abs/#1}{arXiv:#1}}
\usepackage{array}
\newcolumntype{P}[1]{>{\centering\arraybackslash}m{#1}}

\let\Im=\Imag

\DeclareMathOperator{\Op}{Op}

\let\Im=\Imag 

\def\?[#1]{\textbf{[#1]}\marginpar{\Large{\textbf{??}}}}

\def\smallsection#1{\smallskip\noindent\textbf{#1}.}
\let\epsilon=\varepsilon 

\newtheorem{theo}{Theorem}
\newtheorem{prop}{Proposition}[section]

\numberwithin{equation}{section}

\DeclareMathOperator{\Spec}{Spec}

\title[Periodic band inversion]{A mathematical study of periodic band inversion}

\author{Tyler Guo}
\email{tylerguo@berkeley.edu}
\author{Maciej Zworski} 
\email{zworski@berkeley.edu}
\address{Department of Mathematics, University of California, Berkeley, CA 94720}

\begin{document}
\begin{abstract}
    We give a mathematical analysis of the periodic band inversion phenomenon observed by Tan--Devakul for an electron in a two-dimensional periodic potential coupled to a circularly polarized photon cavity mode. In the strong-coupling limit, we derive an effective Bloch Hamiltonian and prove convergence of the low-lying bands. For a cosine potential, we explain the periodic closing and reopening of the first spectral gap, prove the existence and generic persistence of Dirac cones at the gap-closing points, and compute the Chern numbers associated to isolated band clusters. We also show that higher isolated band clusters cannot persist in the small-coupling regime. Finally, we resolve an apparent sign discrepancy between Berry curvature computations and Chern numbers by tracking the descent from the covering space to the Brillouin torus.
\end{abstract}

\maketitle

\section{The Hamiltonian}

We study an interesting Hamiltonian introduced by Tan--Devakul \cite{tri}. 
To describe it we let $ V \in C^\infty ( \mathbb R^2/\Gamma ; \mathbb R ) $ be a $ \Gamma$-periodic potential where we take $ \Gamma = 2 \pi \mathbb Z^2 $. The specific potential considered in \cite{tri} was 
$ V ( x ) =  V_0 ( \cos x_1 + \cos x_2 ) $. 

\begin{figure}
\includegraphics[width=16cm]{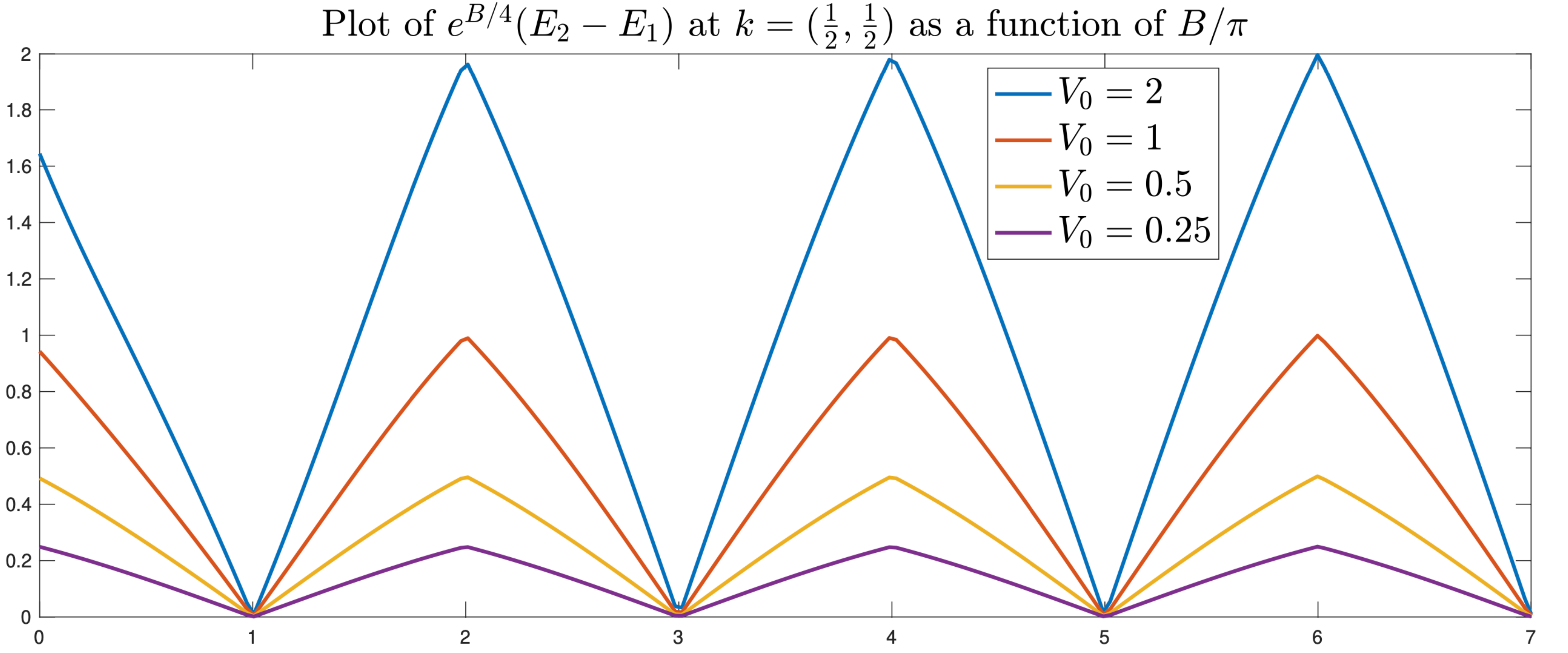}
\caption{\label{f:1} Reinterpretation of \cite[Figure 2(b)]{tri}: after rescaling the gap for the effective Hamiltonian 
\eqref{eq:effH} 
displays an approximate periodicity: $ e^{\frac{B}{4}} ( E_2 - E_1 ) = V_0 ( | |\cos(\frac{B}{4})| - |\sin(\frac{B}{4})||) + 
 \mathcal O ( V_0^2 e^{ - \frac{B}{4}} ) $.  For an animated version see \url{https://math.berkeley.edu/~zworski/gap_movie.mp4}.}
\end{figure}

We use the following notation for $\partial_z $ and annihilation operators:
\[   \mathcal D = 2 D_z = D_{x_1} - i D_{x_2}   , \ \ z = x_1 +  i x_2  , \ \ \
A =  \tfrac1 {\sqrt 2 } ( D_w - i w ) , \ \ w \in \mathbb R , \ \   D = \tfrac 1 i \partial . 
 \]
The Hamiltonian in \cite{tri} is the following self-adjoint operator on $ L^2 ( \mathbb R^2_x \times \mathbb R_w, dx  dw  ) $ which models an electron in a 2D crystal coupled to a circularly polarized photon cavity mode \cite{tri2}:
\begin{equation}
\label{eq:Ham}  H_J := \mathcal D \mathcal D^* + J ( A - \lambda \mathcal D )^* ( A - \lambda \mathcal D) + V ( x) . 
\end{equation}

The operator $A$ describes a single photon mode, with $J$ setting the photon excitation scale.  Following the cavity interpretation in \cite{tri2}, the interaction couples this oscillator to one circular component of the electron momentum.  The choices $2D_z$ and $2D_{\bar z}$ correspond to opposite circular polarizations, or photon helicities.  We choose the $D_z$ helicity so that the signs of the parent Berry curvature and of the induced Chern bands agree directly with \cite{tri}.  For real $V$, complex conjugation intertwines the two helicities (and sends $k$ to $-k$), so their band energies and gaps agree, whereas all Berry curvatures and Chern numbers have opposite signs.

The operator $ H_J $ is $ \Gamma $-periodic in $ z = x_1 + i x_2 $ so it is natural to consider 
Bloch--Floquet spectrum, 
\begin{equation}
\label{eq:defEj} 
\begin{gathered} E_1 ( k, B, J  ) \leq E_2 ( k , B, J  ) \leq \cdots \leq E_n ( k , B , J ) \to \infty ,  \\   k \in \mathbb C/\mathbb Z^2 , \ \ \ 
B = 2 \lambda^2,  \ \ k = k_1 + i k_2 \in \mathbb C ,
\end{gathered} 
\end{equation}
given by the eigenvalues of 
\begin{equation}
\label{eq:Hk}  H_J ( k ) := ( \mathcal D-\bar k  )( \mathcal D -\bar k )^* + J ( A - \lambda (\mathcal D-\bar k  ) )^* ( A - \lambda ( \mathcal D -\bar k  ) ) + V ( x ) . \end{equation}
with periodic boundary conditions in $ x $ with respect to $ \Gamma = 2 \pi \mathbb Z^2 $.

\begin{figure}
\includegraphics[width=14cm]{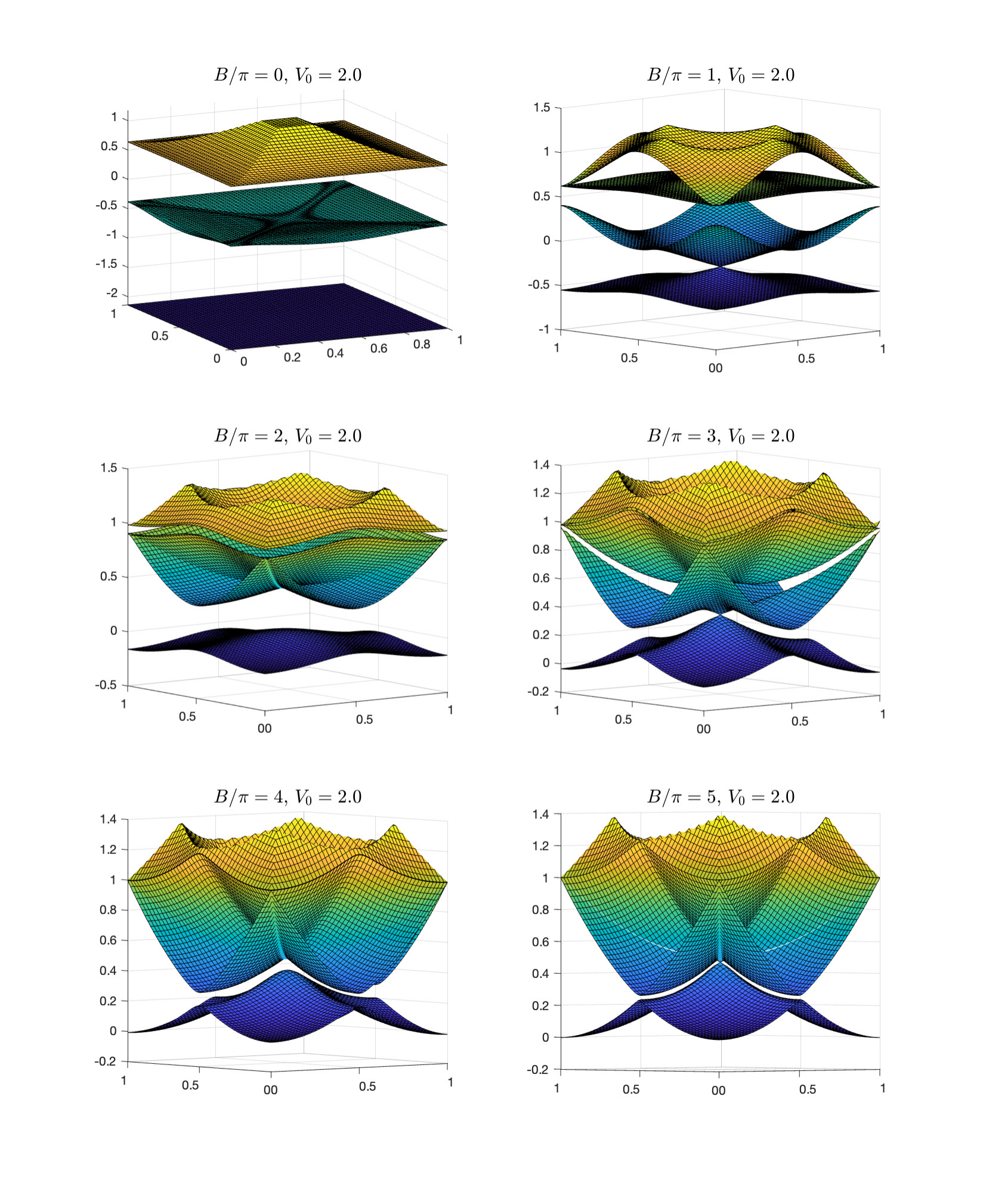}
\caption{\label{f:2} The first four bands for the effective Hamiltonian \eqref{eq:effH}
plotted over the Brillouin torus. Perturbation theory shows that
the bands for \eqref{eq:Hk} are within $ \mathcal O ( 1/J ) $ of the bands of \eqref{eq:effH}. That means that separated bands remain
separated (for large $ J $) and the change of Chern numbers computed in \S \ref{s:chern} shows that they have to touch.
For an animated plot of $ k \mapsto E_j ( (k,k), B ) $ see \url{https://math.berkeley.edu/~zworski/TD_movie.mp4}.}
\end{figure}

Tan--Devakul \cite{tri} considered the case of $ J \to \infty $ and made striking observations about the bands and their topology:  as a function of $ B $ the gap between the bands, 
\begin{equation}
\label{eq:defg}  g ( B ) := \min_{ k \in \mathbb C/\mathbb Z^2 } (E_2 ( k, B  ) - E_1 ( k , B  ) ), \end{equation}
oscillates with regularly spaced zeros and local maxima (see \cite[Figure 2(b)]{tri} where
$ g ( ( 2 \ell + 1 ) \pi ) = 0 $ and with local maxima at $ B = 2\pi \ell  $, $ \ell \in \mathbb N $).
Every time the gap closes, the topology of the line bundle corresponding to $ E_1 ( k , B )$ changes with the Chern number decreasing by one. This process is referred to in physics as {\em band inversion}.
That is done by considering an effective Hamiltonian ($ E_j (k, B ) $ are its bands) which provides an approximation in the $ J \to \infty $ limit.

The purpose of this paper is to clarify and extend some of the results of \cite{tri} and our findings can be 
summarised as follows:
\begin{itemize}

\item Construction of an effective Hamiltonian \eqref{eq:effH} for an arbitrary periodic potential $ V $ and 
the analysis of the $ J \to \infty $ approximation -- see Theorem \ref{t:bands}.

\

\item Perturbative analysis of  the gap $ g ( B ) $ in \eqref{eq:defg}, for  the case $ V ( x ) = V_0 ( \cos x_1 + \cos x_2 ) $ as $ V_0 \, e^{ - \frac{B}{4} } \to 0 $ -- see Figure \ref{f:1} and Theorem \ref{t:gap}.

\

\item Proof of $ g ( ( 2 \ell + 1 ) \pi ) = 0 $, $ \ell \in \mathbb N $ for $ V ( x ) = V_0 ( \cos x_1 + \cos x_2 ) $; showing that the 1st and 2nd bands, and 3rd and 4th bands touch at conic (Dirac) points for 
$ \gamma := V_0 e^{- ( 2 \ell + 1 )\pi/4 }$ small; we also show that 
$ \gamma \in \mathbb R \setminus \mathcal A $, where $ \mathcal A  $ is a discrete set, the Dirac cone structure persists  -- 
see Figure \ref{f:2} and Theorem \ref{t:Dirac}.

\

 \item Proof that there is {\em no} gap for bands corresponding to $ E_j ( k , B) $, $ j > 3 $, for small
 $ |V_{0}| e^{-\frac{B}{4} } $. The gap here is understood as a minimum over $ k $, not the gap between the unions over all $ k $ -- see \eqref{bandisolation1} and \eqref{bandisolation2}. We also estimate the number of bands that can be isolated for an arbitrary $ V_0 e^{-\frac{B}{4} } $ -- see Theorem \ref{t:nogap}.

 \

\item Computation of the Chern number, $  -\ell $, of the line bundle of eigenvectors of $ E_1 ( k , 2 \pi \ell ) $, 
 and of the rank 2-bundle, $  -2 \ell $,  of eigenvectors of $ E_j ( k, 2 \pi \ell ) $, $ j = 2, 3 $ (guaranteed to be separated for {\em small} or {\em large} values 
 $ |V_0| e^{ - \frac{\pi\ell}{2}} $ and expected to be separated throughout) 
 -- see Theorem \ref{t:chern}.

 \end{itemize}

We also address the sign difference emphasized in \cite{tri} -- see \S \ref{s:63}.  With the helicity chosen above, the parent Berry curvature is $+B$, in direct agreement with \cite{tri}. 
The Berry connection on the parent bundle in the case of $ V = 0 $ (see \eqref{eq:etap}) does {\em not} descend to a connection on the dual torus $ \mathbb R^2/\Gamma^* $ under ordinary Bloch identification and there is no well defined Chern number. 
In the special case  $B=2\pi\ell$, magnetic translations give an honest descent of the parent line to a bundle $P$ (see 
\eqref{eq:defP}) with $c_1(P)=\ell$, and the parent Berry connection descends to it.  The physical Bloch bundle for the isolated band created by a non-trivial periodic potential, $L $, has the opposite Chern number $c_1(L)=-\ell$ (see Theorem \ref{t:chern}). Thus the parent (defined using magnetic translations) and potential-induced bands have opposite Chern numbers. However, the parent bundles are {\em not} in any global sense limits of the Bloch bundle $L $ as $ V_0 \to 0 $ -- see Figure \ref{f:curv}.

 We now provide precise statements of results. We start with the effective Hamiltonian. For a real-valued 
 $ 2 \pi \mathbb Z $-periodic potential $ W $ we define an operator on $ L^2 ( \mathbb R^2 ) $ with the domain
 given by $ H^2 ( \mathbb R^2 ) $, 
\begin{equation}
\label{eq:effH}  H = \mathcal D^* \mathcal D + W^{\rm{w}} ( x_1 - B D_{x_2} , x_2 ) , \end{equation}
 where we consider $ x_1 $ as a parameter and (for $ u \in \mathscr S ( \mathbb R^2 ) $) define
\[   W^{\rm{w}} ( x_1 - B D_{x_2} , x_2 ) u := \Op^{\rm{w}}
 ( W ( x_1 - B \xi_2, x_2 ) ) . \]
 where we used Weyl quantization, see \cite[Chapter 4]{z12}, 
 \begin{equation}
 \label{eq:Opw}
\Op^{\rm{w}} ( a ) u := a^{\rm{w}} ( x, D ) u := 
 \frac{1}{ (2 \pi)^2 } \int\!\int a ( \tfrac12 ( x + y  ) , \xi ) e^{ i \langle x - y , \xi \rangle } u ( y) dy d \xi  .
\end{equation}
The operator \eqref{eq:effH} is unitarily equivalent to 
\begin{equation}
    \label{eq:effHt} 
    \begin{split}
 H_\theta & := e^{ i \theta B D_{x_1} D_{x_2} }
H  e^{ - i \theta B D_{x_1} D_{x_2} } \\
& = 
\mathcal D^* \mathcal D + W^{\rm{w}} ( x_1 - (1 - \theta ) B D_{x_2}, 
x_2 + \theta B D_{x_1} ) ,
\end{split}
\end{equation}
so that $ \theta = \frac12 $ is the symmetric gauge.

 The bands for $ H $ in \eqref{eq:effH} are defined as the spectrum, 
 \begin{equation}
 \label{eq:bandskB} E_1 ( k , B ) \leq E_2 ( k, B ) \leq \cdots \leq E_n ( k, B ) \to \infty , \ \ n \to \infty , \ \ \ 
 k = k_1 + i k_2 \in \mathbb C , \end{equation}
of $ H( k ) : H^2 ( \mathbb C/\Gamma ) \to L^2 ( \mathbb C/\Gamma )$, where 
\begin{equation} 
 \label{eq:effHk} 
 H ( k ) = H( k , B ) := (\mathcal  D-\bar k )^* (\mathcal D -\bar k ) + W^{\rm{w}} ( x_1 - B ( D_{x_2} -k _2 ) , x_2 ) .
 \end{equation}
The first result relates the bands of $ H_J $ to the bands of $ H $ with a specific $ W $:

\begin{theo}
\label{t:bands}
For any fixed $ E $ and definitions \eqref{eq:defEj}, \eqref{eq:bandskB} with
\begin{equation}
\label{eq:effV}
W = \exp ( - B \mathcal D^* \mathcal D  /4 ) V , 
\end{equation}
we have, for $ E_j ( k , B , J ) < E $, 
\begin{equation}
\label{eq:J2eff}     E_j ( k , B , J ) = E_j ( k , B ) +  \mathcal O ( 1/J ) , \ \  J \to \infty .
\end{equation}
Here $ \exp ( - t \mathcal D^* \mathcal D ) $ is the heat propagator for the Laplacian $ \Delta = -  \mathcal D^* \mathcal D $.
\end{theo}

The potential considered in \cite{tri} was given by 
 \begin{equation}
\label{eq:V2V0}  V (x ) = V_0 ( \cos x_1 + \cos x_2 ) , 
\end{equation}
so that the potential $ W $ in \eqref{eq:effH} is given by 
\begin{equation}
\label{eq:defW}
W(x ) = V_0 e^{-\frac{B}{4}}  ( \cos x_1 + \cos x_2 ) . 
\end{equation}
In this case we have a perturbative formula for the gap:
\begin{theo}
\label{t:gap}
Suppose that  $ V $ is given by \eqref{eq:V2V0}
 and that $| V_0| e^{ - \frac{B}{4} } \ll 1 $. Then 
with $ g $ in \eqref{eq:defg}, the notation of Theorem \ref{t:bands}, and $ k_0 = ( \frac12, \frac12 ) $,
\begin{equation}
\label{eq:gap}
\begin{split}    g ( B ) & = E_2 ( k_0, B ) - E_1 ( k_0 , B) \\
& = 
 |V_0| e^{  -\frac{B}{4} } ( | |\cos(\tfrac14{B})| - |\sin(\tfrac14{B})||) +  \mathcal O ( V_0^2 e^{ - \frac{B}{2}} ) . \end{split} 
 \end{equation}
 \end{theo}

The remarkable accuracy of the approximation \eqref{eq:gap} is illustrated by Figure \ref{f:1} and animation linked there. 
It suggests that at $ B = ( 2 \ell + 1 ) \pi $, $ \ell \in \mathbb N $, the bands touch. That is indeed the case for all 
$ V_0 $ with the bands touching at a Dirac point except possibly a discrete set of values of $ V_0 $ (numerically it seems to happen for all values). That is seen from perturbation theory for small values of $ V_0 $ and the results about persistence of Dirac points by Drouot--Lyman \cite{DrLy26} (see that paper for references to earlier literature, in particular to the work of Fefferman--Weinstein \cite{feffe}).

\begin{theo}
\label{t:Dirac} 
Suppose that $ V $ is given by \eqref{eq:V2V0} and $ H(k ) $ and $ E_j ( k ,B ) $ are given by 
\eqref{eq:effH},  \eqref{eq:bandskB}.  Then for all $ V_0 $, the eigenvalues of
$ H(k_0 ) $, $ k_0 = (\frac12, \frac12) $ with $ B_0 := ( 2 \ell + 1 ) \pi $, $ \ell \in \mathbb N $ have even dimension and in particular,
\[  E_{1+p } ( k_0, B_0 ) = E_{2 + p}  ( k_0 , B_0 )  
\ \ p = 0 , 2 , 
 \ \  k_0 = ( \tfrac12 , \tfrac12 ) . \]
If $ | V_0 | e^{-\frac{B}{4} } \ll 1 $ then there exist positive definite matrices $ A_p ( \alpha) $ such that for $p =0, 2 $,
\begin{equation}
\label{eq:Dirac}  
\begin{split} & E_{ 2 + p } ( k , B_0 ) - 
E_{1+p}  ( k_0, B_0 )  = 
\langle A_p( V_0 e^{ -\frac{B}{4}}) (k-k_0), k-k_0 \rangle^{\frac12} + 
 \mathcal O ( | k - k_0 |^2 )  ,   \end{split} \end{equation} 
Moreover, there exists a discrete set $ \mathcal A \subset \mathbb R $ such that, for $ \alpha := V_0 e^{ -\frac{B}{4} } \notin \mathcal A $, \eqref{eq:Dirac}
for $ p = p_j$, $ j=1,2$, for some $ 1 \leq p_1 +1 < p_2 $.
\end{theo}

We believe that we can take $ p_1 = 0 $ and $ p_2 = 2$ in the second part of the theorem but, as in \cite{feffe} and \cite{DrLy26}, the general arguments cannot exclude the possibility of the cone "going up". In forthcoming work we plan to show \eqref{eq:Dirac} when $| V_0 | \gg 1 $ (the semiclassical limit). 

We show that the isolation conditions in Theorem 4, at least for small values of $  V_0 e^{-\frac{B}{4} }  $ (but most likely for all values) cannot hold except for the first three eigenvalues:

\begin{theo}
\label{t:nogap}
There exist $ c_0 , c_1 > 0 $
such that if $M\geq c_{0}(1+V_{0}^{2} e^{-\pi\ell})$,
then for any finite set $ \mathcal S \subset \mathbb N \setminus [1, M ] $,
\begin{equation}\label{bandisolation2}
    \inf_{k\in\mathbb{R}^{2}/\mathbb{Z}^{2}}\mathrm{dist}\Big(\{E_j (k, 2 \pi \ell ) :
  j\in    \mathcal S  \},\{E_j (k, 2 \pi \ell ) :  j \notin \mathcal S \}\Big) = 0. 
\end{equation}
Moreover, when $ |V_0 | e^{ -\frac{\pi\ell}{2} } < c_1 $ then \eqref{bandisolation2} holds for every finite $ S \subset \mathbb N \setminus \{ 1, 2, 3 \}. $
\end{theo}

\noindent 
{\bf Remark.} The proof gives $ c_0 $ and $ c_1 $ but we did not attempt to find  optimal constants.

The next result confirms the calculation of the Chern number for the first band from \cite{tri} and provides
a calculation of the Chern number of the vector bundle associated to the next two bands (assuming that they 
are isolated):
\begin{theo}
\label{t:chern}
Suppose that $ g ( 2 \pi \ell ) > 0 $ (which, thanks to \eqref{eq:gap}, holds for small or large $ |V_0| e^{ - \frac{\pi \ell}{2} }  > 0 $). 
Then the Chern number of the line bundle $ L $ associated to the first band (see \cite[(9.2)]{notes}) is given 
by $ c_1 ( L ) = -\ell $. 

Assuming that for $ E_j ( k ) := E_j ( k , 2 \pi \ell ) $, 
\begin{equation}\label{bandisolation1}
    \inf_{k\in\mathbb{R}^{2}/\mathbb{Z}^{2}}\mathrm{dist}\Big(\{E_j (k) :  j = 2,3 \},\{E_j (k) :  j \notin \{ 2,3 \} \}\Big)>0. 
\end{equation}
(valid for $ |V_0 e^{-\frac{\pi \ell}{2} }| > 0  $ small or large), the Chern number of the corresponding rank-2 vector bundle $ V $, 
is given by $ c_1 ( V ) =  -2 \ell $.
\end{theo}

At the moment the detailed analysis is provided in the case of 
$ |V_0| e^{-\frac{\pi\ell}{2} } $ small. The semiclassical case $ |V_0| \to \infty$
will be addressed in detail later. The conclusions about large values of $ |V_0| e^{-\frac{\pi\ell}{2} }$, $ B = 2 \pi \ell $, in Theorem \ref{t:chern} follow from 
the semiclassical analysis of the one dimensional problem $ -h^2 \partial_x^2 + \gamma_0 \cos x $ found in \cite[\S 6.3]{notes}. It implies that for $ |V_0 |e^{ - \frac{\pi\ell}{2}} $ large 
(which is the value of $ h^{-\frac12} $ after rescaling), the first band at $ B = 2 \pi \ell $ is separated from other bands. The same goes for 
the union of 2nd and 3rd band, and so on (governed by the multiplicities of the eigenvalues of the $2D $ harmonic oscillator), as long as 
$ E_m \leq C h^{-1} = C |V_0|^2 e^{-\pi\ell } $ (this is consistent with Theorem \ref{t:nogap}).

\smallsection{Notation} We use the \emph{mathematics}
inner product convention on $L^2 $, $
\langle u,v\rangle:=\int \! u\,\overline v$
which is also the convention used in \cite[(8.20), \S\S 9.1, 9.4]{notes} which compares to the bra-ket notation (on the right) as follows
$ \langle u, v \rangle = \langle v | u \rangle $. The notation $ f = \mathcal O_V ( g)$ or $ A = \mathcal O_{V \to 
W } ( g ) $ means that $ \| f \|_V \leq C |g| $ and $ \| A \|_{V\to W } \leq C| g|$, respectively. 

\smallsection{Acknowledgements}
We acknowledge the support by the Simons Targeted Grant 896630 ``Moir\'e Materials Magic". We are also grateful to Trithep Devakul for 
introducing us to this subject, for translating the notation of \cite{tri} to the PDE notation \eqref{eq:Ham}, and for subsequent discussions. We would like to thank Alexis Drouot for a helpful discussion of the results of \cite{feffe} and \cite{DrLy26}. We also acknowledge further help from ChatGPT in translating physics concepts into mathematical language in \S \ref{s:63} and for the help in proofreading.


\section{Effective Hamiltonian in the $ J \to \infty $ limit.}
\label{s:effH}

In this section we derive the formula for the effective Hamiltonian and prove Theorem \ref{t:bands}.
We start with a unitary transformation which is implicit in \cite[\S III]{tri}. It is based on the observation 
(we denote by $ w $ the operator $ f ( w ) \mapsto w f ( w ) $) 
\[   
D_w - a = e^{  i a w } D_w e^{- i a w }, \ \    w - b = e^{ - i b D_w } w e^{  i b D_w } , 
 \ \ \ e^{ i b D_w } f ( w) = f ( w + b ) , 
 \]
so that 
\[
D_w - i w - ( a + i b ) = e^{i a w } e^{  i b D_w } ( D_w - i w ) e^{ -i b D_w } e^{ - i a w } . 
\]
Putting $ B := 2 \lambda^2 $, 
\[ 
A - \lambda ( \mathcal D -\bar k  ) = \tfrac{1}{\sqrt 2 } ( D_w - i w - \sqrt B ( D_{x_1}  - i D_{x_2} -\bar k  ) ) ,
\]
and motivated by this we define the following family of unitary operators:
\[ 
U_k  := e^{ i \sqrt B  w (D_{x_1}  -k _1 ) } e^{ -i \sqrt B D_w ( D_{x_2} -k _2  )  }  : L^2 ( \mathbb R^2 /\Gamma \times \mathbb R ) \to L^2 ( \mathbb R^2/\Gamma  \times \mathbb R ) ,
\]
$  k = k_1 + i k_2 \in \mathbb C $.  Then
\[  
A - \lambda ( \mathcal D -\bar k  )  = 
\tfrac 1 {\sqrt 2 } U_k  ( D_w - iw ) U_k^*  .  
\]
and 
\[  
\begin{split} 
& U^*_k (| \mathcal D-\bar k|^2 + J ( A - \lambda (\mathcal D-\bar k) )^* ( A - \lambda( \mathcal D-\bar k ) )) U_k  = \\
&  \ \ \ \ \ \ \ \ \ \ 
(\mathcal D -\bar k )^* ( \mathcal D -\bar k  )  + \tfrac12 J ( D_w^2 + w^2 - 1 ) . \end{split} 
\]

\noindent
{\bf Remark.} We can also use a more symmetric unitary operator
described in terms of Weyl quantisation \cite[\S 4.2.4]{z12}:
\begin{equation}
\label{eq:defUkt}
\widetilde U_k
  =
  {\rm{Op}}^{\rm{w}} \left(e^{i\ell_k}\right) = e^{-\frac{i B}{2}k_1k_2} U_k , 
\end{equation}
where, in the notation of \eqref{eq:Opw}, 
\begin{equation*}
  \ell_k(x,w;\xi,\eta)
  =  \sqrt B\left[w(\xi_1-k_1)-\eta(\xi_2-k_2)\right]
  +\tfrac12{B}\left(\xi_1\xi_2-k_2\xi_1-k_1\xi_2\right) .
 \end{equation*}

Denoting by $ V $ the operator $ u ( x , w )  \mapsto V  ( x ) u ( x, w ) $ on 
$ L^2 ( \mathbb R^2 / \Gamma \times  \mathbb R ) $, we now obtain 
\begin{equation}
\label{eq:HamU} \begin{split} U_k^* H_J ( k ) U_k & = \widetilde H_J ( k ) \\
& := (\mathcal D -\bar k )^* ( \mathcal D -\bar k  )  + \tfrac12 J ( D_w^2 + w^2 - 1 ) + 
U_k^* V U_k  . \end{split} \end{equation}
We analyse
$  U_k^* V U_k   $ using the Fourier expansion: $ V(x)=\sum_{n\in\mathbb{Z}^2}\widehat{V}(n)e^{i ( n_1 x_1 + n_2 x_2 )}$, $ \widehat V ( n ) = 
\mathcal O ( \langle n \rangle^{-\infty } ) $, 
where the rapid decay of coefficients follows from the smoothness assumption. 
Then
\begin{align*}
U^{*}_kVU_k &=\sum_{n\in\mathbb{Z}^2}\widehat{V}(n)e^{i n_1 (x_1-\sqrt{B}w-B( D_{x_2}-k_2 ) )}e^{in_2 (x_2+\sqrt{B}D_{w})}.
\end{align*} 
We will define the effective Hamiltonian by projecting  $ \widetilde H_J ( k ) $ to the ground state of 
the harmonic oscillator,  $\psi_{0}(w):=\pi^{-\frac{1}{4}}e^{-\frac{w^{2}}{2}}$:
\[ 
H ( k ) := (D_{x_{1}}-k_{1})^{2}+(D_{x_{2}}-k_{2})^{2}+ \mathcal W (k),\] where
\begin{align*}
\mathcal W (k)&:=\Big\langle \sum_{n\in\mathbb{Z}^2}\widehat{V}(n)e^{i n_1 (x_1-\sqrt{B}w-B(D_{x_2}-k_{2}))}\,e^{in_2(x_2 +\sqrt{B}D_{w})}\psi_{0} , \psi_{0}\Big\rangle_{L^{2}(\mathbb{R}_{w})}
\\
&= \sum_{n\in\mathbb{Z}^2}\widehat{V}(n)e^{i n_1 (x_1 -B(D_{x_2}-k_{2}))}e^{in_2 x_2 }\Big\langle e^{-i n_1\sqrt{B}w}e^{i n_2 \sqrt{B}D_{w}}\psi_{0},\psi_{0}\Big\rangle_{L^{2}(\mathbb{R}_{w})} 
\\
&=\pi^{-\frac{1}{2}}\sum_{n\in\mathbb{Z}^2}\widehat{V}(n)\bigg(\int_{\mathbb{R}}e^{-\frac12 {w^{2}}}e^{-i n_1 \sqrt{B}w}e^{-\frac12{(w+ n_2 \sqrt{B})^{2}}}\mathrm{d}w\bigg)e^{i n_1 (x_{1}-B(D_{x_2}-k_{2}))}e^{in_2 x_2 }
\\
&=\sum_{n\in\mathbb{Z}^2 }\widehat{V}(n)\,e^{-\frac14 B |n|^2 }e^{\frac12 {iBn_1 n_2 } }\,e^{in_1(x_{1}-B(D_{x_2}-k_{2}))}e^{in_2 x_2 }.
\end{align*}
We let  $\xi :=(\xi_1 ,\xi_2 ),$ $x^{*}:=n ,$ $\xi^{*}:=(0, -Bn_1 )$ and define the linear symbol
\[
\ell_{n}(x,\xi):=\langle x^{*}, x\rangle+\langle\xi^{*},\xi \rangle=n_1 x_1 +n_2  x_2 -B n_1 \xi_2.
\]
We have 
\[ 
\begin{split} \Op^{\mathrm{w}}(e^{i\ell_{n}}) & =e^{i\langle x^{*}, x\rangle}e^{\frac{i}{2}\langle x^{*},\xi^{*}\rangle}e^{i\langle\xi^{*}, D_{x}\rangle}=e^{-\frac{i}{2}B n_1 n_2 }e^{i n_1 x_1 }e^{in_2 x_2 }e^{-in_{1}BD_{x_2}}
\\
& =e^{\frac{i}{2}B n_1 n_2 }e^{i n_1 x_1 }e^{-i n_1 BD_{x_2 }}e^{in_2 x_2 }. \end{split} 
\]
If we define $W:=e^{\frac{B}{4}\Delta}V,$ then 
\[ 
\begin{split}
    \mathcal W (k) & =\sum_{n\in\mathbb{Z}^2}\widehat{W}(n)e^{\frac{i}{2}B n_1 n_2 }e^{i n_1 Bk_{2}}e^{in_1(x_1 -BD_{x_2})}e^{in_2 x_2 } 
\\
&=\sum_{n\in\mathbb{Z}^2}\widehat{W}(n)\Op^{\mathrm{w}}(e^{i(n_1x_1+n_2 x_2 -B n_1 (\xi_2 -k _{2}))}) =\Op^{\mathrm{w}}(a_{k}),
\end{split} 
\]
where the symbol $a_{k}(x, \xi ):=W (x_1 -B(\xi_2 -k _{2}), x_2 ).$  This shows that 
\[ H ( k ) = e^{ i \langle x, k \rangle } H e^{ -i \langle x, k \rangle } , \]
where $ H $ is given in \eqref{eq:effH} with $ W $ in \eqref{eq:effV}. We can now give
\begin{proof}[Proof of Theorem \ref{t:bands}]
We put 
\[ 
\begin{gathered} 
\mathcal{H}:=L^{2}(\mathbb{T}^{2}_{x}\times\mathbb{R}_{w})\cong L^{2}(\mathbb{T}^{2})\otimes 
    L^{2}(\mathbb{R}),
    \\
    \mathcal{D}:=H^{2}(\mathbb{T}^{2}_{x};L^{2}(\mathbb{R}_{w}))\cap L^{2}(\mathbb{T}^{2}_{x};\mathcal{D}(D^{2}_{w}+w^{2})), 
    \\
    H_0 :=\tfrac{1}{2}(D^{2}_{w}+w^{2}-1),\ \ \ \widetilde{V}(k):=U^{*}_k VU_k.
\end{gathered} 
\]
Let $\Pi_{0}: L^{2}(\mathbb{R})\ni u\mapsto \langle u , \psi_0 \rangle\psi_0 \in\ker H_0 $ be the orthogonal projection onto the ground state. We define
\begin{equation*}
\begin{gathered}
    \Pi:=I_{L^{2}(\mathbb{T}^{2})}\otimes\Pi_{0},\ \ \Pi^{\perp}:=I_{L^{2}(\mathbb{T}^{2})}\otimes (1-\Pi_{0}),\\R_{+}u:=\Pi u,\ \ R_{-}u_{-}:=-u_{-},
    \end{gathered}
\end{equation*}
and set up the Grushin problem
\begin{equation*}
    \mathcal{Q}^{J}(k, z):=\begin{pmatrix}
        \widetilde{H}_{J}(k)-z & R_{-}\\R_{+} & 0
    \end{pmatrix}: \mathcal{D}\times \mathrm{ran}\,\Pi\rightarrow \mathcal{H}\times\mathrm{ran}\,\Pi. 
\end{equation*}
We fix an arbitrary $E>0$ and assume $z\leq E$ and  $J\rightarrow\infty.$ We have for all $u\in\mathcal{D}$
\[
\langle\widetilde{H}_{J}(k)\Pi^{\perp} u, \Pi^{\perp}u\rangle\geq J\langle H_{0}\Pi^{\perp}u,\Pi^{\perp}u\rangle+\langle\widetilde{V}(k)\Pi^{\perp}u,\Pi^{\perp}u\rangle\geq (J-\|\widetilde{V}(k)\|_{L^{2}\rightarrow L^{2}})\|\Pi^{\perp}u\|^{2}_{L^{2}}.
\]
This implies that $\Pi^{\perp}\widetilde{H}_{J}(k)\Pi^{\perp}-z$ is invertible on $\mathrm{ran}\,\Pi^{\perp}$ and
\[
\|(\Pi^{\perp}\widetilde{H}_{J}(k)\Pi^{\perp}-z)^{-1}\|_{L^{2}\rightarrow L^{2}}\leq (J-\|\widetilde{V}(k)\|_{L^{2}\rightarrow L^{2}}-E)^{-1}=\mathcal O(\frac{1}{J}).
\]
Hence, we can write
\[
\mathcal{Q}^{J}(k, z)^{-1}=\begin{pmatrix}
    E^{J}(k, z) & E^{J}_{+}(k, z)\\E^{J}_{-}(k, z) & \mathcal{F}^{J}_{-+}(k, z)
\end{pmatrix},
\]
where \[
\begin{split}
\mathcal{F}_{-+}^{J}(k, z)&:=(\Pi\widetilde{H}_{J}(k)\Pi-z)-\Pi\widetilde{H}_{J}(k)\Pi^{\perp}(\Pi^{\perp}\widetilde{H}_{J}(k)\Pi^{\perp}-z)^{-1}\Pi^{\perp}\widetilde{H}_{J}(k)\Pi
\\
&\ =(\Pi\widetilde{H}_{J}(k)\Pi-z)-\Pi\widetilde{V}(k)\Pi^{\perp}(\Pi^{\perp}\widetilde{H}_{J}(k)\Pi^{\perp}-z)^{-1}\Pi^{\perp}\widetilde{V}(k)\Pi.
\end{split}
\]
We have
\[
z\in\Spec (\widetilde{H}_{J}(k))\iff 0\in\Spec (\mathcal{F}^{J}_{-+}(k, z))
\]
and the eigenvalues have the same multiplicity.

We identify $\Pi\widetilde{H}_{J}(k)\Pi$ as the effective Hamiltonian $H(k)$ and consider the set\[
\Sigma_{E}:=\{(k,\lambda)\in \mathbb R^2 / \mathbb Z^2 \times[-C, E]\, : \,  \lambda\in\Spec(H(k))\},
\]
where $C>0$ is chosen such that $\Spec(H(k))\subseteq[-C,\infty)$ for all $k.$ For each point $(k_{*},\lambda_{*})\in \Sigma_{E}$, choose a small disc $D_{*}\subseteq\mathbb{C}$ centered at $\lambda_{*}$ such that $\partial D_{*}\cap \,\Spec(H(k_{*}))=\emptyset.$ After shrinking to a neighborhood $U_{*}\ni k_{*}$, $D_{*}$ contains exactly $m_{*}$ eigenvalues of $H(k_{*})$ counted with multiplicity, for every $k\in U_{*}$.
\\\\
We let $\{U_{k}\times (D_{k}\cap\mathbb{R)}\}_{k=1}^{n(E)}$ be a finite open cover of $\Sigma_{E}$ satisfying the conditions above, and set up a local Grushin problem on each $U_{*}.$ We define
\[
\Pi_{*}(k):=-\frac{1}{2\pi i}\int_{\partial D_{*}}(H(k)-\zeta)^{-1}\mathrm{d}\zeta,\ \ \ \mathrm{rank}\,\Pi_{*}(k)=m_{*}\ \mathrm{on}\ U_{*}.
\]
We now choose a local orthonormal frame
\[
\varphi_{1}(k),\dots,\varphi_{m_{*}}(k)\in\mathrm{Ran}\,\Pi_{*}(k),
\]
and define
\[
\begin{gathered}
R_{-,*}(k):\mathbb{C}^{m_{*}}\ni c\mapsto\sum_{j=1}^{m_{*}}c_{j}\varphi_{j}(k)\in L^{2}(\mathbb{T}^{2}),
\\
R_{+,*}(k):H^{2}(\mathbb{T}^{2})\ni u\mapsto(\langle u,\varphi_{1}(k)\rangle,\dots,\langle u,\varphi_{m_{*}}(k)\rangle)\in\mathbb{C}^{m_{*}}.    
\end{gathered}
\]
With this choice, the Grushin problem
\[
\mathcal{P}^{0}_{*}(k, z):=\begin{pmatrix}
    H(k)-z & R_{-,*}(k)\\ R_{+,*}(k) & 0
\end{pmatrix}
\] is well-posed for $z\in D_{*}$ with
\[
\begin{gathered}
    \mathcal{P}^{0}_{*}(k, z)^{-1}=\begin{pmatrix}
     E^{0}_{*}(k, z) &   R_{-,*}(k)
     \\ 
     R_{+,*}(k) & E^{0}_{-+,*}(k, z)
    \end{pmatrix},
    \\
    E^{0}_{*}(k, z):=\big((H(k)-z)\big|_{\mathrm{Ran}(I-\Pi_{*}(k))}\big)^{-1}(I-\Pi_{*}(k)),
    \\
    E^{0}_{-+,*}(k, z):=zI_{\mathbb{C}^{m_{*}}}-\Pi_{*}(k)H(k)\Pi_{*}(k).
\end{gathered}
\]
We define
\[
W^{J}(k, z):=-\Pi\widetilde{V}(k)\Pi^{\perp}(\Pi^{\perp}\widetilde{H}_{J}(k)\Pi^{\perp}-z)^{-1}\Pi^{\perp}\widetilde{V}(k)\Pi,
\]
and replace $H(k)-z$ by
\[
\mathcal{F}_{-,+}^{J}(k, z)=H(k)-z+W^{J}(k, z).
\]
Define the perturbed Grushin problem
\[
\mathcal{P}_{*}^{J}(k, z):=\begin{pmatrix}
    \mathcal{F}_{-+}^{J}(k, z) & R_{-,*}(k)\\ R_{+,*}(k) & 0
\end{pmatrix}.
\]
Since as $J\rightarrow\infty$
\[
\|E^{0}_{*}(k, z)W^{J}(k, z)\|,\|W^{J}(k, z)E^{0}_{*}(k, z)\|\leq \frac{C_{*}}{J}
\]
uniformly for $k\in U_{*}$ and $z\in D_{*}$, \cite[Proposition 2.12]{notes} gives the well-posedness of $\mathcal{P}^{J}_{*}(k, z)$ with
\[
\mathcal{P}^{J}_{*}(k, z)^{-1}=\begin{pmatrix}
 E^{J}_{*}(k, z) & E^{J}_{+,*}(k, z)\\E^{J}_{-,*}(k, z) & E^{J}_{-+,*}(k, z)   
\end{pmatrix},
\]
where
\[
E^{J}_{-+,*}(k, z):=E^{0}_{-+,*}(k, z)+\sum_{l=1}^{\infty}(-1)^{\ell}R_{+,*}(k)W^{J}(k, z)(E_{*}^{0}(k, z)W^{J}(k, z))^{l-1}R_{-,*}(k).
\]
Hence, we have for $z\in D_{*}, \,k\in U_{*}$
\[
z\in\Spec(\widetilde{H}_{J}(k))\iff 0\in\Spec(\mathcal{F}^{J}_{-+}(k, z))\iff 0\in \Spec (E^{J}_{-+,*}(k, z)),
\]
and as $J\rightarrow\infty$
\[
\|E^{J}_{-+,*}(k, z)-(zI_{\mathbb{C}^{m_{*}}}-\Pi_{*}(k)H(k)\Pi_{*}(k))\|\leq\frac{C_{*}}{J}
\]
holds uniformly. Thus, \eqref{eq:J2eff} follows by taking $C$ to be the maximum of $C_{*}$ over finitely many patches.
\end{proof}

\section{Perturbative formula for the gap}
\label{s:pert}

In this section we consider 
the specific potential \eqref{eq:defW} from \cite{tri} using perturbation theory as $ V_0 e^{ - \frac{B}{4} } \to 0 $. 
To prove Theorem \ref{t:gap} we  construct a Grushin problem (see 
the previous section and \cite[\S 2.6]{notes}) for the operator
$P_{k}(B):=P(x, D-k; B)$,  $ P ( x, D; B ) := H  $ where $ H  $ is the effective Hamiltonian 
\eqref{eq:effH} (the notation is chosen for consistency with \cite[Example 13]{notes} which we will use). 

Putting 
$P_{k}:=(D_{x_{1}}-k_{1})^{2}+(D_{x_{2}}-k_{2})^{2}$, we see that
 \[\Spec_{L^{2}(\mathbb{R}^{2}/\Gamma)}P_{k}=\{|m-k|^{2}\, : \,  m\in\mathbb{Z}^{2}\}.\]
Inside the fundamental domain $(-\frac{1}{2},\frac{1}{2}]^{2},$ the first eigenvalue $\lambda_{1}(k)$ is not simple if and only if $k_{1}=\frac{1}{2}$ or $k_{2}=\frac{1}{2}.$ 
More precisely, if we denote the orthonormal eigenfunctions by 
$e_{m}(x):= (2 \pi)^{-1} \exp (i \langle m , x \rangle ) $, 
$ m\in\mathbb{Z}^{2}$,
then we have the following 3 cases: \begin{enumerate}
    \item [(1)]  At the horizontal edge: $-\frac{1}{2}<k_{1}<\frac{1}{2}$, $k_{2}=\frac{1}{2}$, $\lambda_{1}(k)=k_{1}^{2}+\frac{1}{4}$, $$\ker (P_{k}-\lambda_{1}(k))=\mathrm{span}(e_{0,0}, e_{0,1}).$$
    \item [(2)] At the vertical edge: $k_{1}=\frac{1}{2}, -\frac{1}{2}<k_{2}<\frac{1}{2}$, $\lambda_{1}(k)=\frac{1}{4}+k_{2}^{2},$ $$\ker (P_{k}-\lambda_{1}(k))=\mathrm{span} (e_{0,0}, e_{1,0}).$$
    \item [(3)] At the corner: $k_{1}=k_{2}=\frac{1}{2}$, $\lambda_{1}(k)=\frac{1}{2}$, $$\ker (P_{k}-\lambda_{1}(k))=\mathrm{span}(e_{0,0}, e_{1,0}, e_{0, 1}, e_{1,1}).$$
\end{enumerate}

To simplify notation, we denote the perturbation parameter and the effective potential by 
\[ \gamma:=e^{-\frac{B}{4}}V_{0}, \ \ \ V_{B}:=\gamma(\cos(x_{1}-B(D_{x_{2}}-k_{2}))+\cos x_{2}). \]
We first consider the case where $-\frac{1}{2}<k_{1}<\frac{1}{2},$ $k_{2}=\frac{1}{2}$, and 
define
\[R_{-}: \mathbb{C}^{2}\ni 
u_-\mapsto u_{-,1}e_{0,0}+u_{-,2}e_{0, 1},  \ \ \  R_+ := R_-^* :H^{2}(\mathbb{R}^{2}/\Gamma) \to \mathbb{C}^{2}. \]
(Here and below we take the standard inner product of $ \mathbb C^n $.)

For a sufficiently small $\epsilon>0$,  $|z-k_{1}^{2}-\frac{1}{4}|<\epsilon$, $ |\gamma|<\epsilon$, implies that 
the Grushin problem for $ P_{ (k_1, \frac12)} -z  $ with these $ R_\pm $
is well posed and, with the standard notation for the inverse,
\begin{equation}
    \label{eq:EEp}
 \begin{gathered} Ev:=\sum_{m\neq (0,0), (0, 1)}((m_{1}-k_{1})^{2}+(m_{2}-\tfrac{1}{2})^{2}-z)^{-1}\langle v, e_{m}\rangle e_{m}, \\
E_- = R_+, \ \ \ E_+ = R_- , \ \ \ E_{-+}:=(z-k_{1}^{2}-\tfrac{1}{4})I_{\mathbb{C}^{2}}. \end{gathered} \end{equation}
The perturbed Grushin problem, with $ P_{(k,1/2)} $ replaced by $P_{(k,1/2) } ( B ) $, is also well posed. The corresponding $ E_{-+}^\gamma $ term is given by \cite[Proposition 2.12]{notes} using 
\eqref{eq:EEp}:
\begin{equation}\label{neumann}
E^{\gamma}_{-, +}:=E_{-+}+\sum_{l=1}^{\infty}(-1)^{\ell}R_{+}V_{B}(EV_{B})^{l-1}R_{-}.    
\end{equation} 
From\[\cos(x_{1}-B(D_{x_{2}}-k_{2}))e_{m_{1}, m_{2}}=\tfrac{1}{2}(e^{-iB(m_{2}-k_{2})}e_{m_{1}+1, m_{2}}+e^{iB(m_{2}-k_{2})}e_{m_{1}-1, m_{2}}),\]\[\cos(x_{2})e_{m_{1}, m_{2}}=\tfrac{1}{2}(e_{m_{1}, m_{2}+1}+e_{m_{1}, m_{2}-1}),\] 
we obtain 
\[E_{-, +}^{\gamma}=\begin{pmatrix}
    z-k_{1}^{2}-\tfrac{1}{4} & -\tfrac 12 \gamma  \\ -\tfrac12 \gamma  & z-k_{1}^{2}-\tfrac{1}{4}
\end{pmatrix}+ \mathcal O(\gamma^{2}) .\]
Since $ E_j ( (k_1, 1/2 );B) $, $ j = 1, 2 $ are the values of $ z $ for which 
$ \det E^\gamma_{-+} ( z) = 0 $ we obtain
\[ E_2 ((k_{1},\tfrac{1}{2});B )-E_{1}((k_{1}, \tfrac{1}{2});B)=|\gamma|+\mathcal{\mathcal O}(\gamma^{2}).\]
A similar argument applies to the case $k=(\frac{1}{2}, k_{2}),$ and gives
\[
E^{\gamma}_{-,+}=\begin{pmatrix}
z-k_{2}^{2}-\tfrac{1}{4} & -\tfrac12 \gamma 
e^{-iBk_{2}}\\-\tfrac12 \gamma e^{iBk_{2}} & z-k_{2}^{2}-\tfrac{1}{4}
\end{pmatrix}+\mathcal \mathcal O(\gamma^{2}),
\]
so that 
\[E_{2}((\tfrac{1}{2},k_{2});B )-E_{1}((\tfrac{1}{2}, k_{2});B)=|\gamma|+\mathcal{\mathcal O}(\gamma^{2}).\]

The interesting case is given by $k_0 :=(\frac{1}{2},\frac{1}{2}).$ We define 
\begin{equation}
    \label{eq:defRpm} 
    \begin{split}
& R_{-} :\mathbb{C}^{4}\ni 
u_- 
\mapsto u_{-,1}e_{0,0}+u_{-,2} e_{1,0}+u_{-,3}e_{0, 1}+u_{-,4} e_{1,1}\in L^{2}(\mathbb{R}^{2}/\Gamma), \\
& R_+ = R_-^* : H^2 ( \mathbb R^2/\Gamma ) \to \mathbb C^4 .
\end{split}
\end{equation}
If  $z\in B(\frac{1}{2}, \epsilon)$, $|\gamma|<\epsilon$, $ 
0 < \varepsilon \ll 1 $, the corresponding Grushin problem
for $ P_{k_0} $ is well posed, and
$E: L^{2}(\mathbb{R}^{2}/\Gamma)\rightarrow H^{2}(\mathbb{R}^{2}/\Gamma),$ $E_{-+}:\mathbb{C}^{4}\rightarrow\mathbb{C}^{4}$ are given by 
$$Ev := \sum_{\substack{m \neq (0,0),(0,1)\\ (1,0),(1,1)}} 
\bigl((m_1-\tfrac12)^2+(m_2-\tfrac12)^2-z\bigr)^{-1}\,\langle v,e_m\rangle\, e_m,\ \ E_{-+}:=(z-\tfrac{1}{2})I_{\mathbb{C}^{4}}.$$ 
The perturbed Grushin problem, with $ P_{k_0} $ replaced by $P_{k_0}  ( B ) $, is also well posed and the corresponding $ E_{-+}^\gamma $ term is given by the analogue of \eqref{neumann}.
A computation then shows that
\[E^{\gamma}_{-, +}=(z-\tfrac{1}{2})I_{\mathbb{C}^{4}}-\tfrac12 \gamma \begin{pmatrix}
    0 & e^{-\frac{iB}{2}} & 1 & 0\\e^{\frac{iB}{2}} & 0 & 0 & 1\\1 & 0 & 0& e^{\frac{iB}{2}}\\0 & 1 & e^{-\frac{iB}{2}} & 0
\end{pmatrix}+\mathcal{\mathcal O}(\gamma^{2} ),\]
and that the spectrum of the 4-by-4 matrix above is given by 
$\pm 2\cos(\frac14{B}) $, $ \pm 2\sin(\frac14{B}) $. 

Hence, up to $\mathcal \mathcal O(V_{0}^{2}e^{-\frac{B}{2}})$, the first four eigenvalues of $P_{k_0} (B)$  are given by 
$$\left\{\tfrac{1}{2}\pm V_{0} e^{-\frac{B}4} \cos(\tfrac14{B}), \tfrac{1}{2}\pm V_{0} e^{-\frac{B}4} \sin(\tfrac14{B}) \right\}.$$ 
This gives \eqref{eq:gap} proving Theorem \ref{t:gap}. 


\section{Dirac cones at $ B = ( 2 \ell + 1 ) \pi $}
\label{s:Dirac}

This section is devoted to the proof of Theorem \ref{t:Dirac}. We first show that bands always touch 
when $ B $ is an odd multiple of $ \pi $ (as already suggested by \eqref{eq:gap}). Then 
we show that for small values of $ V_0 e^{ - ( 2 \ell + 1 ) \pi/4 }$ the first two bands touch at a Dirac point. Finally, we apply results of \cite{DrLy26} to see that Dirac points persist for all but (possibly) a discrete set of coupling constants. 

\subsection{Symmetry induced degeneracy}

Let $ k_0 = ( \frac12 , \frac12 ) $. We claim that for the operator \eqref{eq:effHk} with the potential 
\eqref{eq:defW}, 
\begin{equation}
\label{eq:evend}    \forall \, E \in \mathbb R \ \ 
\dim  \ker_{ H^2 ( \mathbb R^2/\Gamma )} ( H ( k_0, ( 2 \ell + 1 )\pi  ) - E ) \in 2 \mathbb N . 
\end{equation}
To see this we introduce the following two unitary operators defined on 
$  L^{2}(\mathbb{R}^{2}/\Gamma) $:
\begin{equation}
\label{eq:defS}
(S_{0}u)(x_{1}, x_{2}):=e^{ix_{2}}u(x_{1}+\pi,-x_{2}),\ \ \ \ (S_{1}u)(x_{1}, x_{2}):=e^{i(x_{1}+x_{2})}u(-x_{1}, -x_{2}).
 \end{equation}
We then check that 
\[  S_j^2 = I, \ \ \ [  H ( k_0, (2\ell+1)\pi) , S_j ] = 0 , \ \ \ j = 0,1 , \]
and that
\[ S_1 S_0 = - S_0 S_1 .\]
This shows that every eigenspace of $ H ( k_0 , (2\ell+1)\pi)$ is even dimensional.
In fact, since eigenvalues of $ S_j $ are $ \pm 1 $, for any $ E $,
\[  V_E:= \ker ( H ( k_0, (2 \ell+1)\pi ) - E ) = V_E^+ \oplus V_E^- , \ \ 
V_E^\pm := \{ u \in V_E : S_0 u = \pm u \} . \]
But then $ S_1 : V_E^\pm \to V_E^\mp $ and hence $ \dim V_E^+ = \dim V_E^- $. 

\subsection{Perturbative calculations near $(2\ell+1)\pi$}
\label{s:Diracc} 
We let $\gamma:=V_{0}e^{-\frac{B}{4}}$,  $ R_\mp $ be as in \eqref{eq:defRpm}. Then 
for  $|z-\frac{1}{2}|<\epsilon,$ $|k-(\frac{1}{2},\frac{1}{2})|<\epsilon$, 
$ \varepsilon \ll 1 $ the corresponding Grushin problem for 
$ P_k - z $ is well posed. In the inverse $ E_\pm = R_\mp $, and
$$Ev := \sum_{\substack{m \neq (0,0),(0,1)\\ (1,0),(1,1)}} 
\bigl((m_1-k_{1})^2+(m_2-k_{2})^2-z\bigr)^{-1}\,\langle v,e_m\rangle\, e_m,$$
$$E_{-+}:= z I_{\mathbb C^4 } - \begin{pmatrix}
    k_{1}^{2}+k_{2}^{2} & & & \\ & (1-k_{1})^{2}+k_{2}^{2} & & \\ & & k_{1}^{2}+(1-k_{2})^{2} & \\ & & & (1-k_{1})^{2}+(1-k_{2})^{2}
\end{pmatrix}.$$

Using \eqref{neumann}, we obtain
\begin{align*}
    E^{\gamma}_{-,+}&=E_{-,+}-\frac{\gamma}{2}\begin{pmatrix}
    0 & e^{-iBk_{2}} & 1 & 0\\e^{iBk_{2}} & 0 & 0 & 1\\1&0&0&e^{iB(1-k_{2})}\\0 & 1& e^{-iB(1-k_{2})} & 0
\end{pmatrix}+\mathcal \mathcal O(\gamma^{2})\\&=:\widetilde{E}^{\gamma}_{-,+}+\mathcal \mathcal O(\gamma^{2}).
\end{align*}
We let $B_{0}:=(2\ell+1)\pi$, $\kappa:=k-(\frac{1}{2},\frac{1}{2}),$\ $\delta:=B-B_{0}.$ 
A direct calculation shows
\begin{align*}
    \det(\widetilde{E}^{\gamma}_{-,+}(z, k, B))=0\iff z\in \{\widetilde{z}_{1}(\kappa,B),\widetilde{z}_{2}(\kappa,B), \widetilde{z}_{3}(\kappa,B), \widetilde{z}_{4}(\kappa,B)\},
\end{align*}
where $\widetilde{z}_{1}\leq \widetilde{z}_{2}\leq \widetilde{z}_{3}\leq \widetilde{z}_{4}$ are given by
\[\widetilde{z}_{1,2,3,4}(\kappa, B)=\tfrac{1}{2}+|\kappa|^{2}\pm \sqrt{|\kappa|^{2}+\tfrac{\gamma^{2}}{2}\pm2\sqrt{\kappa_{1}^{2}\kappa_{2}^{2}+\tfrac{\gamma^{2}}{4}|\kappa|^{2}+\tfrac{\gamma^{4}}{16}\sin^{2}(\tfrac{\delta}{2}) }}.\]

Hence, if $|\gamma|$ is sufficiently small and fixed, then near $\kappa=0$ and $B=B_{0}$ the first four eigenvalues of $P_{k}(B)$ have the following approximate formulas
\begin{equation}\label{appformula}
z ( \kappa, B ) = \tfrac12 \pm 2^{-\frac12}  \gamma \pm
2^{-\frac12} \sqrt{ |\kappa|^2 + \tfrac14 \gamma^{2} \sin^{2}(\tfrac{\delta}{2}) } + \mathcal{O}_{\gamma}(|\kappa|^2 + \delta^{2})+\mathcal O (\gamma^{2}).  
\end{equation}


\subsection{Dirac cones for small values of $V_{0}e^{-\frac{B}{4}}$}
We fix
\[
k_{0}:=(\tfrac{1}{2},\tfrac{1}{2}),\ B_{0}:=(2\ell+1)\pi,\ \gamma:=V_{0}e^{-\frac{B_{0}}{4}}.
\]
Assume $0<|\gamma|\ll1$, so that by \eqref{eq:evend}, \eqref{appformula} the first eigenvalue $E_{0}$ of $H(k_{0}, B_{0})$ has multiplicity 2.
To justify the first two bands $E_{1, 2}(k, B_{0})$ touch at a Dirac point, we need to remove the correction term $\gamma^{2}$ from \eqref{appformula}. More precisely, we prove
\begin{prop}
    Let $\rho:=|\kappa|+|\delta|$ be sufficiently small. There exist constants $a,b,c,\alpha,\beta\in\mathbb{R}$ such that the first two eigenvalues of $H(k, B)$ are given by the approximate formula
    \begin{equation}
    \label{diracconesforsmallgamma}
    E_{1, 2}(k, B)=E_{0}+a\delta\pm\sqrt{(\alpha\kappa_{1}+\beta\kappa_{2})^{2}+c^{2}\kappa_{2}^{2}+b^{2}\delta^{2}}+\mathcal{O}(\rho^{2}).
    \end{equation}
\end{prop}


\begin{proof}
    We define
\[
\begin{split}
W(k, B)&:=H(k, B)-H(k_{0}, B_{0})\\&\ =\kappa_{1}K_{1}+\kappa_{2}K_{2}+\delta Q+\mathcal{O}_{H^{2}\rightarrow L^{2}}(\rho^{2}),
\end{split}
\]
where
\[
\begin{gathered}
    K_{1}:=-2(D_{x_{1}}-\tfrac{1}{2}),\\
    K_{2}:=-2(D_{x_{2}}-\tfrac{1}{2})-\gamma B_{0}\sin (x_{1}-B_{0}(D_{x_{2}}-\tfrac{1}{2})),
    \end{gathered} \]
    and
    \[   \begin{split} Q &:=-\gamma(Q_{1}+Q_{2})\\
    & = -\gamma\Big(\tfrac{1}{4}\big(\cos(x_{1}-B_{0}(D_{x_{2}}-\tfrac{1}{2}))+\cos x_{2}\big)-(D_{x_{2}}-\tfrac{1}{2})\sin (x_{1}-B_{0}(D_{x_{2}}-\tfrac{1}{2}))\Big). \end{split}
\]


Let $\Pi_{0}$ be the spectral projector onto $\ker (H(k_{0}, B_{0})-E_{0})$. By \eqref{eq:evend}, we can choose an orthonormal basis $\{\psi_{+},\psi_{-}\}\subseteq\mathrm{Ran}\,\Pi_{0}$ such that
\begin{equation}
\label{psipm}
 S_{0}\psi_{\pm}=\pm\psi_{\pm},\ \ S_{1}\psi_{\pm}=\psi_{\mp}.   
\end{equation}

We put
\[
R_{+}:H^{2}(\mathbb{R}^{2}/\Gamma)\ni u\mapsto \begin{pmatrix}
    \langle  u,\psi_{+}\rangle\\\langle u, \psi_{-}\rangle
\end{pmatrix}\in\mathbb{C}^{2},\]\[R_{-}:\mathbb{C}^{2}\ni \begin{pmatrix}
    c_{+}\\c_{-}
\end{pmatrix}\mapsto c_{+}\psi_{+}+c_{-}\psi_{-}\in L^{2}(\mathbb{R}^{2}/\Gamma).
\]
If we let $|z-E_{0}|+|\delta|+|\kappa|<\epsilon\ll 1,$ then the Grushin problem
\[
\begin{pmatrix}
    H(k, B)-z & R_{-}\\R_{+} & 0
\end{pmatrix}:H^{2}(\mathbb{R}^{2}/\Gamma)\times\mathbb{C}^{2}\rightarrow L^{2}(\mathbb{R}^{2}/\Gamma)\times\mathbb{C}^{2}
\]
is invertible with
\[
\begin{gathered}
E_{-+}(z, k, B)=(z-E_{0})I_{\mathbb{C}^{2}}-\kappa_{1}A_{1}-\kappa_{2}A_{2}-\delta M+\mathcal{O}(\rho^{2}),\\A_{1}:=R_{+}K_{1}R_{-},\ A_{2}:=R_{+}K_{2}R_{-},\ M:=R_{+}QR_{-}.
\end{gathered}
\]
We denote
\[
\sigma_{1}:=\begin{pmatrix}
    0 & 1\\1 & 0
\end{pmatrix},\ \ \ \sigma_{2}:=\begin{pmatrix}
    0 & -i\\i & 0
\end{pmatrix},\ \ \ \sigma_{3}:=\begin{pmatrix}
    1 & 0\\0& -1
\end{pmatrix}.
\]
A computation shows that
\[
[K_{1}, S_{0}]=0, K_{1}S_{1}=-S_{1}K_{1},
\]
\[
(D_{x_{2}}-\tfrac{1}{2})S_{0}=-S_{0}(D_{x_{2}}-\tfrac{1}{2}),\ (D_{x_{2}}-\tfrac{1}{2})S_{1}=-S_{1}(D_{x_{2}}-\tfrac{1}{2}),
\]
\[
[\sin(x_{1}-B_{0}(D_{x_{2}}-\tfrac{1}{2})), S_{0}]=0, \]
\[ \sin(x_{1}-B_{0}(D_{x_{2}}-\tfrac{1}{2}))S_{1}=-S_{1}\sin(x_{1}-B_{0}(D_{x_{2}}-\tfrac{1}{2})), \]
\[
[Q_{1}, S_{0}]=[Q_{1}, S_{1}]=[Q_{2}, S_{1}]=0, \ Q_{2}S_{0}=-S_{0}Q_{2}. 
\]
It follows that there exist $\alpha,\beta,a,b,c\in\mathbb{R}$ such that
\[
A_{1}=:\alpha\sigma_{3},\ A_{2}=:c\sigma_{2}+\beta\sigma_{3},\ M=:a I_{\mathbb{C}^{2}}+b\sigma_{1}.
\]
Hence, the two eigenvalues of $E_{-+}(z,k,B)$ are given by the approximate formula
\[
z-E_{0}-a\delta\pm\sqrt{(\alpha\kappa_{1}+\beta\kappa_{2})^{2}+c^{2}\kappa_{2}^{2}+b^{2}\delta^{2}}+\mathcal{O}(\rho^{2}).
\]
This completes the proof of the Proposition.
\end{proof}

Taking $\delta=0$ in \eqref{diracconesforsmallgamma}, we obtain
\begin{equation}
\label{diracconesforsmallgamma2}
E_{1,2}(k, B_{0})=E_{0}\pm\sqrt{(\alpha\kappa_{1}+\beta\kappa_{2})^{2}+c^{2}\kappa_{2}^{2}}+\mathcal{O}(|\kappa|^{2}).   
\end{equation}
To prove Theorem \ref{t:Dirac} for $|\gamma|$ sufficiently small, it remains to show
\begin{equation}
\label{diracconesforsmallgamma3}
\alpha=\langle-2(D_{x_{1}}-\tfrac{1}{2})\psi_{+},\psi_{+}\rangle\neq 0\ \ \mathrm{and}\ \ c=i\langle-2(D_{x_{2}}-\tfrac{1}{2})\psi_{-},\psi_{+}\rangle\neq 0   \end{equation}

We use Fourier modes to approximate the eigenfunctions $\psi_{\pm}$ defined in \eqref{psipm} to compute the constants $\alpha,c\in\mathbb{R}$ up to leading order in $\gamma.$

We denote 
\[
\begin{split}
 P_{k_{0}}:=(D_{x_{1}}-\tfrac{1}{2})^{2}+(D_{x_{2}}-\tfrac{1}{2})^{2},\ \ H:=H(k_{0}, B_{0}),
\\
V:=\cos(x_{1}-B_{0}(D_{x_{2}}-\tfrac{1}{2}))+\cos x_{2}. 
\end{split}
\]
Let $\Pi$ be the projector onto $K_{0}:=\mathrm{Ker}(P_{k_{0}}-\frac{1}{2})$ and $\Pi^{\perp}:=1-\Pi.$ Since $\psi_{\pm}$ solve
\[
(H-E_{0})\psi_{\pm}=0
\]
writing $\psi_{\pm}=\Pi\psi_{\pm}+\Pi^{\perp}\psi_{\pm}=:u_{\pm}+u_{\pm,\perp}$ gives
\begin{equation}
\label{approximation}
F(E_{0})u_{\pm}=0,\ \ u_{\pm,\perp}=-\gamma(\Pi^{\perp}(H-E_{0})\Pi^{\perp})^{-1}\Pi^{\perp}V\Pi u_{\pm},
\end{equation}
where
\begin{equation}
\label{approximation2}
\begin{split}
F(E_{0})&=\Pi(H-E_{0})\Pi-\Pi H\Pi^{\perp}(\Pi^{\perp}(H-E_{0})\Pi^{\perp})^{-1}\Pi^{\perp}H\Pi\\&=(\tfrac{1}{2}-E_{0})I_{K_{0}}+\gamma\Pi V\Pi+\mathcal{O}(\gamma^{2}).
\end{split}
\end{equation}

We decompose
\[\begin{gathered}
K_{0}=K_{0,+}\oplus K_{0,-},\ \ K_{0,\pm}:=\{u\in K_{0}\mid S_{0}u=\pm u\}=\mathrm{span}(f_{1,\pm}, f_{2,\pm}),\\f_{1,\pm}=\tfrac{e_{0,0}\pm e_{0, 1}}{\sqrt{2}},\ \ f_{2,\pm}=\tfrac{e_{1, 0}\mp e_{1,1}}{\sqrt{2}}.
\end{gathered}
\]
We let $\Pi_{\pm}$ be the projector onto $K_{0,\pm}$ and take 
\[
v_{\pm}:=\frac{1}{\sqrt{1+(\sqrt{2}\mp 1)^{2}}}(f_{2,\pm}+(-1)^{\ell}i(\sqrt{2}\mp 1)f_{1,\pm})\in K_{0,\pm}
\]
to be the normalized eigenvectors of $\Pi_{\pm}V\Pi_{\pm}$ for the lowest eigenvalue $\lambda=-\frac{\sqrt{2}}{2}.$ Then \eqref{approximation}, \eqref{approximation2} give
\[
\begin{split}
|\alpha|=2|\langle(D_{x_{1}}-\tfrac{1}{2})v_{+},v_{+}\rangle|+\mathcal{O}(\gamma)=\tfrac{\sqrt{2}}{2}+\mathcal{O}(\gamma),\\ |c|=2|\langle (D_{x_{2}}-\tfrac{1}{2})v_{+}, v_{-}\rangle|+\mathcal{O}(\gamma)=\tfrac{\sqrt{2}}{2}+\mathcal{O}(\gamma). \end{split}
\]

\subsection{Persistence of Dirac cones} 
We denote
\[
H_{\gamma}(k):=(D_{x_{1}}-k_{1})^{2}+(D_{x_{2}}-k_{2})^{2}+\gamma\big(\cos(x_{1}-B_{0}(D_{x_{2}}-k_{2}))+\cos x_{2}\big).
\]

\begin{prop}
There exist real analytic families $ \gamma \mapsto \mu_\pm (\gamma)   $
of eigenvalues of $ H_\gamma ( k_{0}) $,
such that near $ \gamma = 0 $,
\[
\mu_{\pm}(\gamma)=\tfrac{1}{2}\pm\tfrac{\gamma}{\sqrt{2}}+\mathcal O(\gamma^{2}).
\]
and a discrete set $D_{\pm}\subseteq\mathbb{R}$ such that for   $ \gamma\in\mathbb{R}\backslash D_{\pm}:$
\begin{enumerate}
    \item [(1)] $\mu_{\pm}(\gamma)$ is an eigenvalue of $H_{\gamma}(k_{0})$ of multiplicity exactly 2;
    \item [(2)] for $E_{\pm, 1}(k,\gamma)\leq E_{\pm,2}(k,\gamma)$, eigenvalues of $H_{\gamma}(k)$ such that
    \[
    E_{\pm, 1}(k_{0},\gamma)=E_{\pm,2}(k_{0},\gamma)=\mu_{\pm}(\gamma),
    \]
we have 
\begin{equation*}
    \begin{split}
        E_{\pm, 1}(k_{0}+\kappa,\gamma)&=\mu_{\pm}(\gamma)-\sqrt{\langle A_\pm ( \gamma ) \kappa, \kappa \rangle } +\mathcal O(|\kappa|^{2}),\\E_{\pm, 2}(k_{0}+\kappa,\gamma)&=\mu_{\pm}(\gamma)+ \sqrt{\langle A_\pm ( \gamma ) \kappa, \kappa \rangle } +\mathcal O(|\kappa|^{2}),
    \end{split}
\end{equation*}
where the real analytic families of 2-by-2 matrices, $ \gamma \mapsto A_\pm (\gamma) $ are positive definite for $ \gamma \notin D_\pm $.
\end{enumerate}
\end{prop}

\begin{proof}
The family 
\[
\gamma\mapsto H_{\gamma}(k_{0})
\]
satisfies the assumption of \cite[Proposition 2.1]{DrLy26}. This gives discrete sets $D_{\pm}\subseteq\mathbb{R}$ and analytic families of orthogonal projections
\[
\Pi_{\pm}(\gamma): L^{2}(\mathbb{R}^{2}/\Gamma)\rightarrow L^{2}(\mathbb{R}^{2}/\Gamma)
\]
of constant rank, such that for $\gamma\notin D_{\pm},$ $\Pi_{\pm}(\gamma)$ is the spectral projector of $H_{\gamma}(k_{0})$ associated to $\mu_{\pm}(\gamma),$ and $\mu_{\pm}(\gamma)$ has constant multiplicity on $\mathbb{R}\backslash D_{\pm}.$

Since every eigenspace of $H_{\gamma}(k_{0})$ is even dimensional, the perturbative approximations at the end of \S \ref{s:Diracc} 
for small $\gamma$ and then the smoothness of $ \Pi_\pm ( \gamma ) $ shows
\[
\mathrm{dim}\,\mathrm{Ran}\,\Pi_{\pm}(\gamma)=2,\ \ \ \gamma\in\mathbb{R}\ \ \ \ \mathrm{and}\ \ \ \ \mathrm{mult}_{H_{\gamma}(k_{0})}(\mu_{\pm}(\gamma))=2, \ \ \ \gamma\notin D_{\pm}.
\]
This proves part (1) of the proposition.

We now fix $\gamma\notin D_{\pm}$ so that $\mu_{\pm}(\gamma)$ is an isolated eigenvalue of $H_{\gamma}(k_{0})$ of multiplicity 2. We choose contour $\Gamma_{\pm}$ around $\mu_{\pm}(\gamma)$ so that
$ 
\Gamma_{\pm}\cap\Spec H_{\gamma}(k_{0})=\emptyset$. 
For $|\kappa|$ sufficiently small, we then have 
$ \Gamma_{\pm}\cap\Spec H_{\gamma}(k_{0}+\kappa)=\emptyset$.
Using this fact we define 
\[
P_{\gamma,\pm}(\kappa):=-\frac{1}{2\pi i}\int_{\Gamma_{\pm}}(H_{\gamma}(k_{0}+\kappa)-z)^{-1}\,\mathrm{d}z,
\]
and choose a local analytic orthonormal basis $u_{1,\pm}(\kappa), u_{2,\pm}(\kappa)$ of $\mathrm{Ran}\,P_{\gamma,\pm}(\kappa).$ (That the rank of the
projection $ P_{\gamma , \pm } ( \kappa ) $ is $ 2 $ follows from the fact that 
$ P_{\gamma, \pm } ( 0) = \Pi_\pm ( \gamma ) $ and the continuity in $ \kappa $.) Let $C_{\pm}(\gamma, \kappa)$ be the $2\times 2$ matrix representing $P_{\gamma,\pm}(\kappa)H_{\gamma}(k_{0}+\kappa)P_{\gamma,\pm}(\kappa)$ defined by
\[
\big(C_{\pm}(\gamma, \kappa)\big)_{ij}:=\Big\langle H_{\gamma}(k_{0}+\kappa)u_{j,\pm}(\kappa), u_{i,\pm}(\kappa)\Big\rangle.
\]
Taylor expansion gives
\[
C_{\pm}(\gamma, \kappa)=\mu_{\pm}(\gamma)I_{\mathbb{C}^{2}}+L_{\pm}(\gamma,\kappa)+\mathcal O(|\kappa|^{2}),
\]
where
\[
L_{\pm}(\gamma,\kappa)=\kappa_{1}\Pi_{\pm}(\gamma)\partial_{k_{1}}H_{\gamma}(k_{0})\Pi_{\pm}(\gamma)+\kappa_{2}\Pi_{\pm}(\gamma)\partial_{k_{2}}H_{\gamma}(k_{0})\Pi_{\pm}(\gamma).
\]
If $l_{\pm}(\gamma,\kappa):=\frac{1}{2}\mathrm{tr}\,L_{\pm}(\gamma,\kappa),$\ $q_{\pm}(\gamma,\kappa):=l_{\pm}(\gamma,\kappa)^{2}-\mathrm{det}\,L_{\pm}(\gamma,\kappa) \geq 0 $ then 
the eigenvalues of $L_{\pm}(\gamma,\kappa)$ are given by
$ 
l_{\pm}(\gamma,\kappa)\pm\sqrt{q_{\pm}(\gamma,\kappa)}. $
Moreover, a direct calculation shows that for $ S_1$ in \eqref{eq:defS},
\begin{equation}\label{S1}
S_{1}H_{\gamma}(k)S_{1}^{-1}=H_{\gamma}(1-k),
\end{equation}
Differentiating \eqref{S1} at $k_{0}=(\frac{1}{2},\frac{1}{2})$ gives
\[
S_{1}\Pi_{\pm}(\gamma)\partial_{k_{j}}H_{\gamma}(k_{0})\Pi_{\pm}(\gamma)S_{1}^{-1}=-\Pi_{\pm}(\gamma)\partial_{k_{j}}H_{\gamma}(k_{0})\Pi_{\pm}(\gamma)\ \ \ j=1,2,
\]
which implies $l_{\pm}(\gamma,\kappa)=0.$ 

To finish the proof of the proposition, it suffices to show the positive definiteness of $q_{\pm}(\gamma,\kappa)$ fails only on a discrete set. 
To see this, we write $q_{\pm}(\gamma,\kappa)=\kappa^{T}G_{\pm}(\gamma)\kappa.$ The function 
$
\gamma\mapsto\det G_{\pm}(\gamma)
$ 
is real analytic on $\mathbb{R}$.  \eqref{diracconesforsmallgamma2}, \eqref{diracconesforsmallgamma3} and similar computations for $E_{3,4}(k, B_{0})$ show that $\det G_{\pm}(\gamma)$ is not identically zero.\end{proof}

\section{Overlaps for higher eigenvalues}
\label{s:overlaps}

In this section we prove Theorem \ref{t:nogap}. 
We let $B=2\pi \ell,$ $l\in\mathbb{Z}$ and $\gamma:=e^{-\frac{\pi \ell}{2}}V_{0}.$
We have \begin{align*}
H_{\gamma}(k)&=(D_{x_{1}}-k_{1})^{2}+(D_{x_{2}}-k_{2})^{2}+\gamma\Big(\cos(x_{1}-2\pi \ell(D_{x_{2}}-k_{2}))+\cos x_{2}\Big)\\&= (D_{x_{1}}-k_{1})^{2}+\gamma\cos(x_{1}+2\pi \ell k_{2})+(D_{x_{2}}-k_{2})^{2}+\gamma\cos x_{2}.
\end{align*}

Thus, \[\Spec (H_{\gamma}(k))=\{E^{\gamma}_{m,n}(k_{1}, k_{2})\, : \,  m,n\in\mathbb{N}_{+}\}=\{\epsilon^{\gamma}_{m}(k_{1})+\epsilon^{\gamma}_{n}(k_{2})\, : \,  m,n\in\mathbb{N}_{+}\},\]
where $\epsilon^{\gamma}_{j}(\kappa)$ denotes the $j$-th eigenvalues of the 1D operator $(D_{t}-\kappa)^{2}+\gamma\cos t.$ 
We let $r\geq 2$ and set 
\begin{equation*}
    A_{r}:=\{(m,n)\in\mathbb{N}^{2}_{+}\, : \,  m+n=r\}.
\end{equation*}

\begin{prop}
\label{overlapsforsmallgamma}
If $|\gamma|$ is sufficiently small, then a subset $S\subseteq\mathbb{N}^{2}_{+}$ satisfying 
\begin{equation}
\label{bandisolation}
    |S|<\infty\ \ \mathrm{and} \ \inf_{k\in\mathbb{R}^{2}/\mathbb{Z}^{2}}\mathrm{dist}\Big(\{E^{\gamma}_{m,n}(k)\, : \,  (m,n)\in S\},\{E^{\gamma}_{m,n}(k)\, : \,  (m,n)\notin S\}\Big)>0.
\end{equation}
implies
\begin{equation}
\label{nogap}
S\subseteq A_{2}\cup A_{3}.
\end{equation}    
\end{prop}

\begin{proof}
If $0\leq \kappa\leq\frac{1}{2},$ then
\begin{equation*}
    \epsilon^{0}_{2a-1}(\kappa)=(a-1+\kappa)^{2},\ \ \ \epsilon_{2a}^{0}(\kappa)=(a-\kappa)^{2},\ \ a\in\mathbb{N}_{+}.
\end{equation*}
If we define $d_{j}(\kappa):=\epsilon^{0}_{j+1}(\kappa)-\epsilon^{0}_{j}(\kappa),$ then
\begin{equation*}
    d_{2a-1}(\kappa)=(2a-1)(1-2\kappa),\ \ d_{2a}(\kappa)=4a\kappa.
\end{equation*}
Hence, if we set 
\begin{equation}\label{sigma}
    \sigma_{j}:=
\left\{
\begin{array}{l}
0,\ j\ \mathrm{odd,} \\
\frac{1}{2},\ j\ \mathrm{even,}
\end{array}
\right.
\end{equation}
then
\begin{equation}
\label{dj}
    d_{j}(\sigma_{j})=j,\ \ d_{j}(\tfrac{1}{2}-\sigma_{j})=0.
\end{equation}

Suppose $m\geq 2$. Then\begin{equation*}
    \begin{split}
        E^{0}_{m,n}(k_{1},k_{2})-E^{0}_{m-1,n+1}(k_{1}, k_{2})&=(\epsilon^{0}_{m}(k_{1})-\epsilon^{0}_{m-1}(k_{1}))-(\epsilon_{n+1}^{0}(k_{2})-\epsilon^{0}_{n}(k_{2}))\\&=d_{m-1}(k_{1})-d_{n}(k_{2}).
    \end{split}
\end{equation*}
We have by \eqref{dj}
\begin{equation}
\label{bandswithsame m+n}
\begin{split}
E^{0}_{m,n}(\sigma_{m-1},\tfrac{1}{2}-\sigma_{n})-E^{0}_{m-1, n+1}(\sigma_{m-1},\tfrac{1}{2}-\sigma_{n})=m-1>0,\\E^{0}_{m,n}(\tfrac{1}{2}-\sigma_{m-1},\sigma_{n})-E^{0}_{m-1,n+1}(\tfrac{1}{2}-\sigma_{m-1},\sigma_{n})=-n<0.    
\end{split}
\end{equation}
Therefore, for $\gamma=0$, if $S$ satisfies condition \eqref{bandisolation} and contains one element of $A_{r},$ then $A_{r}\subseteq S.$

Now suppose $r\geq 4.$ Consider $(1,r-1)\in A_{r}$ and $(3,r-2)\in A_{r+1}.$ For $k_{1}=0,$ \begin{equation*}
    \begin{split}
        E^{0}_{1,r-1}(0,\kappa)-E^{0}_{3,r-2}(0,\kappa)&=(\epsilon^{0}_{r-1}(\kappa)-\epsilon^{0}_{r-2}(\kappa))-(\epsilon^{0}_{3}(0)-\epsilon_{1}^{0}(0))\\&=d_{r-2}(\kappa)-1.
    \end{split}
\end{equation*} Using \eqref{sigma}, we have
\begin{equation}
\label{bandswithdifferent m+n}
   \begin{split}
        E^{0}_{1,r-1}(0,\sigma_{r-2})-E^{0}_{3,r-2}(0,\sigma_{r-2})=r-3>0,
        \\
        E^{0}_{1,r-1}(0, \tfrac{1}{2}-\sigma_{r-2})-E^{0}_{3,r-2}(0,\tfrac{1}{2}-\sigma_{r-2})=-1<0.
   \end{split}
\end{equation}
This together with \eqref{bandswithsame m+n} shows that every $E^{0}_{m,n}$ with $m+n\geq 4$ belongs to an infinite chain of intersections. The estimate
\begin{equation}
\label{separation}
    \sup_{\kappa\in\mathbb{R}/\mathbb{Z}}|\epsilon^{\gamma}_{j}-\epsilon^{0}_{j}|\leq|\gamma|,\ \ j\in\mathbb{N}_{+}
\end{equation}
shows the same conclusion remains true for $|\gamma|$ sufficiently small. 
\end{proof}

\begin{prop}
\label{overlap}
If $\gamma\in\mathbb{R}$ is fixed and 
\[ 
m+n\geq 8\lceil4|\gamma|\rceil+12, 
\]
then every $E^{\gamma}_{m,n}$ belongs to an infinite chain of intersections.
That is, there is an infinite sequence $ ( m_j , n_j ) $ such that
 $ (m,n) = (m_0, n_0) $, $(m_{j},n_{j})\neq (m_{k}, n_{k})$ if $j\neq k$, and 
 \[ 
 \forall \, j \in \mathbb N, \ \ \ E^{\gamma}_{m_j,n_j} (k) = E^{\gamma}_{m_{j+1},n_{j+1}} (k)
 \]
 for some $k\in\mathbb{R}^{2}/\mathbb{Z}^{2}.$
\end{prop}
\begin{proof}
    We use the notation
    \[
    (m,n)\sim(m',n')
    \]
    to mean that $E^{\gamma}_{m,n}$ and $E^{\gamma}_{m',n'}$ intersect somewhere on $\mathbb{R}^{2}/\mathbb{Z}^{2}.$ We let $K:=\lceil 4|\gamma|\rceil+1$, $H:=2K+1.$
    \ 

    We claim that 
    \begin{equation}
    \label{5.3.1}
    (m,n)\sim (m-1, n+1)\ \mathrm{whenever\ }\min (m-1, n)\geq K.
    \end{equation}
    \ 

    Indeed, \eqref{bandswithsame m+n} and \eqref{separation} give
    \[
 \begin{split}
    E^{\gamma}_{m,n}(\sigma_{m-1}, \tfrac{1}{2}-\sigma_{n})-E^{\gamma}_{m-1, n+1}(\sigma_{m-1}, \tfrac{1}{2}-\sigma_{n})\geq (m-1)-4|\gamma|>0,
     \\
     E^{\gamma}_{m,n}(\tfrac{1}{2}-\sigma_{m-1}, \sigma_{n})-E^{\gamma}_{m-1, n+1}(\tfrac{1}{2}-\sigma_{m-1}, \sigma_{n})\leq 4|\gamma|-n<0.
 \end{split}   
    \]
    Thus, the two bands intersect.

    We now show every $E^{\gamma}_{H, N}$ with $N\geq K+H+1$ belongs to an infinite chain of intersections. We compute
    \[
    \begin{split}
    E^{\gamma}_{H, N}(0,\sigma_{N-1})-E^{\gamma}_{H+2,N-1}(0,\sigma_{N-1})\geq N-1-H-4|\gamma|>0,
    \\
    E^{\gamma}_{H, N}(0, \tfrac{1}{2}-\sigma_{N-1})-E^{\gamma}_{H+2,N-1}(0, \tfrac{1}{2}-\sigma_{N-1})\leq 4|\gamma|-H<0.
    \end{split}
    \]
    This shows $(H, N)\sim (H+2, N-1).$ Using \eqref{5.3.1}, we obtain an infinite chain
    \[
    (H, N)\sim (H+2, N-1)\sim (H+1, N)\sim (H, N+1)\sim\cdots\sim(H, N+2)\sim\cdots.
    \]
It remains to connect every $(m,n)$ with $r:=m+n\geq 4H$ to a pair of the form $(H, N).$ 
   
   If  $\min(m, n)\geq K,$ then \eqref{5.3.1} shows
   \[
   (m,n)\sim\cdots\sim(H, m+n-H).
   \]
   Now suppose $1\leq m\leq K.$ We compute
   \[
   \begin{split}
  E^{0}_{m,n}(k_{1}, k_{2})-E^{0}_{H,r-H}(k_{1}, k_{2})&=(\epsilon^{0}_{n}(k_{2})-\epsilon^{0}_{r-H}(k_{2}))-(\epsilon^{0}_{H}(k_{1})-\epsilon_{m}^{0}(k_{1}))\\&\geq\frac{(r-K-1)^{2}-(r-H)^{2}}{4}-\frac{H^{2}}{4}\geq K.     
   \end{split}
   \]
   Hence,
   \[
E^{\gamma}_{m,n}(k_{1}, k_{2})-E^{\gamma}_{H, r-H}(k_{1}, k_{2})\geq K-4|\gamma|>0.  
   \]
  For $s$ sufficiently large, we have $E^{\gamma}_{H, r-H+s}> E^{\gamma}_{m,n}>E^{\gamma}_{H,r-H}.$ Thus,  $E^{\gamma}_{m,n}$ intersects
  \[
  (H, r-H)\sim\cdots\sim(H, r-H+s).
  \]The case $1\leq n\leq K$ follows by symmetry. 
\end{proof}
Proposition \ref{overlap} shows if $
S\subseteq\mathbb{N}^{2}_{+}$ satisfies \eqref{bandisolation}, then
\[
S\subseteq \bigcup_{r=2}^{4H-1}A_{r}.
\]
We then obtain
\[
|S|\leq\sum_{r=2}^{8K+3}(r-1)\leq 32K^{2}+20K+3= \mathcal O (1+|\gamma|^{2}),
\] which completes the proof of Theorem \ref{t:nogap}.

\section{Chern numbers at $ B = 2 \pi \ell $.}
\label{s:chern}

We will now consider the case of $ B = 2 \pi \ell $ and compute the Chern number of the line bundle of 
Bloch--Floquet eigenfunctions corresponding to the lowest (isolated) band.   It is given by $-\ell $ and agrees with the result presented in \cite[\S III]{tri}. Our argument is similar to the 
computation of the Chern number in Thouless pumping \cite[\S 9.1]{notes}.

\subsection{The first band}
\label{s:61}

We consider the {\em smallest} Bloch eigenvalue,  $ \varepsilon ( \kappa, \gamma ) $, for the one dimensional 
potential $ \gamma \cos t $. From \cite[Theorem 6, (10.1)]{notes} we know that there exists an eigenstate
$ u \in C^\infty ( \mathbb R) $ such that
\begin{equation}
\label{eq:defu} \begin{gathered}
 ( ( D_t - \kappa ) ^2 +
 \gamma \cos t ) u ( \kappa , t ) = \varepsilon ( \kappa, \gamma ) u ( \kappa , t ) , \\ 
u ( \kappa +  p , t + 2\pi \ell ) = e^{ i p t } u ( \kappa , t ) , \ \  
p, \ell \in \mathbb Z , \ \ \ \int_0^{2 \pi}  | u ( \kappa , t  ) |^2 dt = 1.   \end{gathered} \end{equation}
Expressing the operator using the Fourier expansion in $ x_2 $ we see that 
\[ \cos ( x_1- 2 \pi \ell ( D_{x_2} -k _2 ) ) = \cos ( x_1 + 2 \pi \ell k_2 ) .  \]
It then follows from \eqref{eq:effH} that
\[ E_1 ( k , 2 \pi \ell ) =  \varepsilon ( k_1 , e^{ - \frac{\pi\ell}{2} } V_0 ) + \varepsilon ( k_2 , e^{ - \frac{\pi\ell}{2} } V_0 ),\]
with the eigenfunction given by 
\begin{equation}
\label{eq:defwk} w ( k, x ) = u ( k_1, x_1 + 2 \pi \ell k_2 ) u ( k_2, x_2 ) ,  \ \ x = ( x_1, x_2 ) , \ \ k = (k_1, k_2 ). \end{equation}
The Chern number of the Bloch--Floquet line bundle, $ L $,  (see \cite[(9.2)]{notes} and \eqref{eq:defL} below) is given by \cite[(9.7)]{notes}:
\begin{equation}
\label{eq:chern1} c_1 ( L ) = \frac{ i } { 2 \pi} \int_{ \mathbb R^2/\Gamma^* } d\eta , 
\ \ \ \eta = \langle d_k w(k,\bullet) , w(k,\bullet) \rangle_{ L^2 ( \mathbb R^2/ \Gamma )} . 
\end{equation}
The curvature $ \Theta := d \eta $ is a well defined 2-form, the 1-form
$ \eta $ is defined on $ F:= [0,1) \times [0, 1)$, a fundamental domain of 
$ \Gamma^* = \mathbb Z^2 $. That gives, 
\[ c_1 ( L ) = \frac{ i } { 2 \pi } \int_{\partial F } \eta ,\]
where the boundary of the square, $ \partial F $ is oriented counter-clockwise. We have
\[ \begin{gathered}  \partial_{k_1 } w ( k , x ) = 
\partial_{\kappa} u ( k_1,  x_1 + 2 \pi \ell k_2 ) u ( k_2, x_2 ) , \\ 
 \partial_{k_2 } w ( k , x ) = 2 \pi \ell  \partial_{t} u  ( k_1,  x_1 + 2 \pi \ell k_2 ) u ( k_2, x_2 )
 + u ( k_1,  x_1 + 2 \pi \ell k_2 ) \partial_{\kappa } u ( k_2, x_2 ). 
 \end{gathered} 
 \]
 Only the first term in $ \partial_{k_2} w $ contributes to \eqref{eq:chern1} as the other term and $ \partial_{k_1} w $ give the Chern number of the trivial bundle $ k \mapsto \mathbb C u ( k_1, x_1 ) u (k_2, x_2 ) $ (which is simple to verify directly). 
Hence, 
\[ \begin{split} c_1 ( L ) & = i \ell \int_{\partial F } \left( \int_{0}^{2 \pi } 
\partial_t u ( k_1, t ) \overline{ u ( k_1, t )} dt \right)  dk_2  \\
& = i \ell \int_0^{2 \pi } \left(
\partial_t u ( 1, t ) \overline{ u ( 1, t )} - \partial_t u ( 0 , t ) 
\overline{ u (0, t ) } \right) dt .
\end{split} \]
We now use \eqref{eq:defu} to see that
\[    u ( 1 , t ) = e^{ i t } u ( 0 , t ) , \ \ \
\partial_t u ( 1, t ) = \partial_t ( e^{ i t } u ( 0 , t ) ) = 
i e^{ it } u ( 0, t ) + e^{it } \partial_t u ( 0, t ) . \]
Since $ \int_0^{2 \pi } | u ( 0 , t ) |^2 dt = 1$ we conclude that
\begin{equation}
\label{eq:Chern}
c_1 ( L ) =  -\ell .
\end{equation}

\subsection{Rank-2 bundle corresponding to $ E_2 ( k , 2 \pi \ell ) $ and $ E_3 ( k , 2 \pi \ell ) $.}

The computations in \cite[Example 13]{notes} show that for $\gamma>0$ sufficiently small, $A_{2}, A_{3}$ satisfy condition \eqref{bandisolation}. We therefore denote $V \rightarrow\mathbb{R}^{2}/\mathbb{Z}^{2}$ the Bloch bundle associated to bands $\{E_{1,2}(k),E_{2,1}(k)\}$. The Chern number of $V $ is given by
\begin{equation*}
    \begin{split}
        c_{1}(V)&:=\frac{i}{2\pi}\int_{\mathbb{R}^{2}/\mathbb{Z}^{2}}\mathrm{tr}\,\Theta =\frac{1}{\pi}\int_{\mathbb{R}^{2}/\mathbb{Z}^{2}}\mathrm{Im}\langle\partial_{k_{1}}w_{1,2}(k,\bullet),\partial_{k_{2}}w_{1,2}(k,\bullet)\rangle \\
        & \hspace{4cm} +\mathrm{Im}\langle \partial_{k_{1}}w_{2, 1}(k,\bullet),\partial_{k_{2}}w_{2,1}(k,\bullet)\rangle\,{d}k_{1} {d}k_{2}.
    \end{split}
\end{equation*}

We have, for each $m,n$
\[\begin{split}
    \langle\partial_{k_{1}}w_{m,n},\partial_{k_{2}}w_{m,n}\rangle&= 2\pi \ell\langle\partial_{\kappa}u_{m},\partial_{t}u_{m}\rangle_{L^{2}_{x_{1}}}+\langle\partial_{\kappa} u_{m}, u_{m}\rangle_{L^{2}_{x_{1}}}\langle u_{n},\partial_{\kappa}u_{n}\rangle_{L^{2}_{x_{2}}},
\end{split}\]and 
\eqref{eq:defu} shows that the second term is 
real-valued. Computing this as 
in the previous section gives $ c_1 ( V) = -2 \ell $. 

\subsection{Comparison with the parent bundle of \cite{tri}}
\label{s:63}
Here we discuss the computation in \cite[\S II.B]{tri} and the relation to 
\S \ref{s:chern}. We use the notation of that section and 
\S \ref{s:effH} and the line-bundle convention of \cite[\S 9.1, (9.5)--(9.7)]{notes}. 
In particular, for a local normalized
frame $\Phi(k)$ of a Hermitian line bundle  write the connection and curvature as 
\[
\eta :=\langle d_k\Phi,\Phi\rangle,
\qquad
\Theta:=d\eta.
\]
Thus
\begin{equation}
    \label{eq:defBk}
 \begin{gathered}
\Theta 
= - i B ( k ) dk_1 \wedge dk_2 , \ \ \
B ( k ) := 2 \Im\langle \partial_{k_1}\Phi,\partial_{k_2}\Phi\rangle,
\\
c_1(L)=\frac{i}{2\pi}\int_{\mathbb R^2/\mathbb Z^2}\Theta 
=
\frac1 { 2\pi} \int_{\mathbb R^2/\mathbb Z^2} B ( k ) dk_1dk_2,
\end{gathered} \end{equation}
where we choose the standard orientation of the torus. We recall that although $ \eta $ is locally defined, 
the curvature 
$ \Theta $ is a 2-form on $ \mathbb R^2/\mathbb Z^2$.

  The conjugation \eqref{eq:HamU} shows that the spectrum of \eqref{eq:Ham} with $ V \equiv 0 $ is given by 
 the union of  $ [ 0 , \infty ) + J m $, $ m = 0, 1, \cdots $. As $ J \to \infty $, the bottom of the spectrum is given by a {\em parent band}, $ |k|^2 $, with generalised eigenfunctions of $ \widetilde H_J ( k ) $ in 
 \eqref{eq:HamU} given by $ ( 2 \pi )^{-1} 
 e^{ i \langle k, x \rangle } \psi_0 ( w ) $, $ k  \in \mathbb R $. After the conjugation \eqref{eq:Hk}, the generalised eigenfunctions are given by $ U_k ( (2 \pi)^{-1} \otimes \psi_0 ) $. 
This suggests the following 
definition of the {\em parent line bundle} which we will compare below to the bundle implicit in 
\cite{tri}. For $k=(k_1,k_2)\in\mathbb R^2$ we use the symmetric unitary transformation \eqref{eq:defUkt} to define
\[
\Phi_0(k):= \frac{1}{ 2 \pi } \widetilde U_k(1\otimes \psi_0),
\qquad
\Phi_0(k,x,w)= 
(2\pi)^{-1}
e^{ -i\sqrt{B}\,k_1 w - \tfrac{iB}{2}k_1 k_2} 
\psi_0 (w + \sqrt{B}\,k_2).
\]
We also note that with the $ L^2 (\mathbb T^2_x \times\mathbb R_w ) $
inner product, 
\begin{equation}  
\label{eq:ff} \langle \Phi_0 ( k ) , \Phi_0 ( k' ) \rangle= \exp \left(
-\tfrac14  {B} \left( | k - k'|^2 
+ 2 i \sigma ( k, k' )  \right) \right), \end{equation}
where $ \sigma ( k , k' ) :=  k_1' k_2 - k_1 k_2' $. This is the helicity for which the overlap phase, and hence the sign of the parent Berry curvature, agrees directly with the parent-state convention of \cite[\S II.B]{tri}. 
From the explicit formula for $\Phi_0 (k) $ we obtain in the notation 
of \eqref{eq:defBk}, 
\begin{equation}
\label{eq:etap}
\begin{gathered}
\eta_{\rm par}=\langle d_k\Phi_0,\Phi_0\rangle=-\tfrac{iB}{2}(k_{1}dk_{2}-k_2\,dk_1),
\\
\Theta_{\rm par}=-iB\,dk_1\wedge dk_2,  \ \ \ \ B ( k ) = B .
\end{gathered}
\end{equation}
This is well defined connection on the (trivial) {\em parent} line bundle over 
$ \mathbb R^2$, $ k \mapsto \Phi_0 ( k ) $ but it does {\em not} descend to a connection on the dual torus. 

To define a line bundle for which the connection descends for special values of $ B $ we use {\em magnetic translations}. For $n=(n_1,n_2)\in\mathbb Z^2$ we define them by
\begin{equation}
\label{eq:Tmn-parent-newer}
(T_nu)(x,w):=e^{-i\sqrt B\,n_1w}u(x,w+\sqrt B\,n_2).
\end{equation}
Then
\begin{equation}
\label{eq:Phi0-equivariance-newer}
\begin{gathered}
\Phi_0(k+n)=g_n(k)T_n\Phi_0(k), \qquad
g_n(k):=e^{\frac{iB}{2}(k_{1}n_{2}-k_{2}n_{1}-n_{1}n_{2})},
\\
T_nT_{n'}=e^{-iBn_1'n_2}T_{n+n'} = e^{-iB( n_1'n_2 - n_1 n_2') } T_{n'}T_n .
\end{gathered}
\end{equation}
Hence for
\[
B=2\pi\ell,
\qquad \ell\in\mathbb Z,
\]
the projective multiplier in the second line of
\eqref{eq:Phi0-equivariance-newer} is trivial, so that
$T_nT_{n'}=T_{n+n'}$.  Hence, the lines
$ \widetilde P_k:=\mathbb C\Phi_0(k)$,
$k\in\mathbb R^2/\mathbb Z^2$,
 define a Hermitian line bundle
$P\to\mathbb R^2/\mathbb Z^2$, as in \cite[(9.2)]{notes}:
\begin{equation}
\label{eq:defP} \begin{gathered}  P := \{ ( k ,v ) : k \in \mathbb R^2 , v \in \mathbb C \Phi_0 ( k ) \}/\sim , \\ ( k, v ) \sim ( k', v' ) \ \Longleftrightarrow 
\ \exists \, n \in \mathbb Z^2 , \ k' = k + n, \ v' = T_n v . 
\end{gathered} \end{equation}
The class $[(k,\Phi_0(k))]$ does not define a section of $P$: 
from \eqref{eq:Phi0-equivariance-newer} and 
\eqref{eq:defP},
\[
\begin{split}
[(k+n,\Phi_0(k+n))]
&=g_n(k)[(k+n,T_n\Phi_0(k))]\\
&=g_n(k)[(k,\Phi_0(k))].
\end{split}
\]
Thus $g_n(k)$ are the transition functions of the local frame
$\Phi_0$.  Moreover,
\[
\eta_{\rm par}(k+n)-\eta_{\rm par}(k)
 =-\tfrac12 {iB} (n_1\,dk_2-n_2\,dk_1)
 =g_n(k)^{-1}d g_n(k),
\]
which is precisely the transformation law for a unitary connection.
Consequently, $\eta_{\rm par}$ does define a connection on $P$, with
curvature $\Theta_{\rm par}$, and
\[
c_1(P)=\frac{i}{2\pi}
\int_{\mathbb R^2/\mathbb Z^2}\Theta_{\rm par}
=\frac{B}{2\pi}=\ell.
\]

\begin{figure}
\includegraphics[width=16cm]{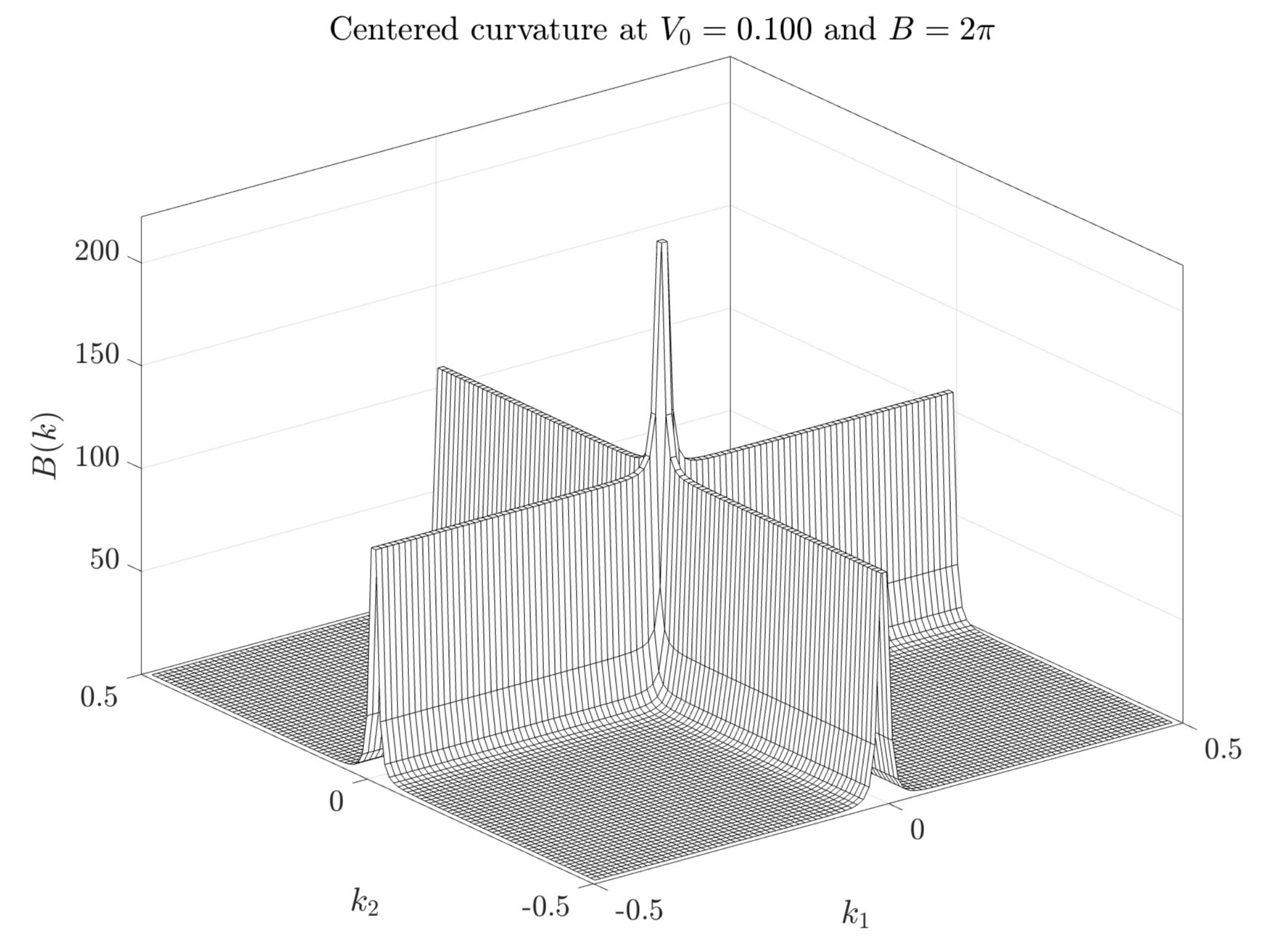}
\caption{\label{f:curv} 
The curvature, $ B ( k ) $, of the line bundle $ L $ in \eqref{eq:defL} for $ B = 2 \pi $ and $ V_0 = 0.1 $ -- the sign is reversed as the 
computation was done for the opposite helicity (that is with $ D_{\bar z } $ in place of $ D_{z } $ in \eqref{eq:Ham}). We use
the symmetric gauge here: \eqref{eq:effHt} with $ \theta = \frac12$
illustrating \eqref{eq:cc}. For an animated version see \url{https://math.berkeley.edu/~zworski/curv_movie.mp4}.}
\end{figure}

We next compare this with the first band of the effective Hamiltonian. 
In the notation of \eqref{eq:defwk} we put
\[
\widetilde L_k:=\mathbb Cw(k,\bullet)\subset L^2(\mathbb T^2_x),
\qquad
\Phi(k):=\widetilde U_k(w(k,\bullet)\otimes\psi_0).
\]
The line bundle considered in \S \ref{s:61} is given by 
\begin{equation}
\label{eq:defL} 
\begin{gathered} 
L := \left\{ ( k , v ) : k \in \mathbb R^2, v \in \mathbb C w(k,\bullet) \right\}/ \sim ,  \\
( k , v ) \sim ( k', v' ) \ \Longleftrightarrow \ 
\exists \, n \in \mathbb Z^2 , \ k' = k +n , \ v' (x )  = e^{i \langle  x, n \rangle } v ( x ) .
\end{gathered}
\end{equation}
(See \cite[Lemma 9.1]{notes} for a proof that this defines a line bundle.) 
This is  the genuine spectral line bundle over the Brillouin torus, while $\Phi(k)$ is its lift to the $U_k$-gauge over $ \mathbb R^2_k$.
The families $\Phi_0(k)$ and $\Phi(k)$ should be
compared on $\mathbb R^2$, while the line bundle obtained from the
Bloch descent on the effective side is $L$.

An additional point should be emphasized when comparing the bundles defined by $ \Phi_0 ( k) $ and $ \Phi ( k ) $. The eigenfunctions 
$ w ( k , x ) $ {\em cannot} be chosen so that $ w ( k ,x ) \to 
( 2 \pi)^{-1} $ as $ V_0 \to 0 $. The Bloch periodicity condition 
prevents that from happening. For the symmetric gauge 
($ \theta = \frac12 $ in \eqref{eq:effHt}), the curvature of $ L $, $ B_{V_0 } (k)$
(see \cite[\S 9.1]{notes} and \eqref{eq:defBk}) can be shown to satisfy
\begin{equation}
\label{eq:cc}B_{V_0 } ( k ) \to  -\pi \ell \sum_{ n \in \mathbb Z^2 } \left( 
\delta ( k_1 + n_1 ) + \delta( k_2 + n_2) \right), \end{equation}
in the sense of distributions on $ \mathbb R^2 $. We illustrate this in Figure~\ref{f:curv}.

\end{document}